\documentclass[12pt,a4paper]{article}

\usepackage{amsmath,amssymb,amsthm}
\usepackage{graphicx}
\usepackage{enumitem}
\usepackage[ruled, vlined]{algorithm2e}
\usepackage[authoryear,round]{natbib}
\usepackage{xcolor,caption}
\usepackage[colorlinks, allcolors = blue, breaklinks=true, hypertexnames=false]{hyperref}

\newcommand{\N}{\mathbb{N}}

\newcommand{\R}{\mathbb{R}}
\newcommand{\E}{\mathbb{E}}
\renewcommand{\P}{\mathbb{P}}

\newcommand{\Deltam}{\Delta_{m-1}}
\newcommand{\Deltamn}{\add{\Omega_{m,n}}}
\newcommand{\barDeltamn}{\add{\bar\Omega_{m,n}}}

\DeclareMathOperator*{\argmin}{arg\,min}
\DeclareMathOperator*{\argmax}{arg\,max}

\newtheorem{Thm}{Theorem}%[section]
\newtheorem{Prop}[Thm]{Proposition}
\newtheorem{Lemma}[Thm]{Lemma}
\newtheorem{Cor}[Thm]{Corollary}

\urldef{\SoccerPred}\url{https://projects.fivethirtyeight.com/soccer-predictions/}

\newcommand{\add}[1]{%{\color{blue}
	#1}%}

\allowdisplaybreaks

\begin{document}

\title{\bf A Simple Algorithm for Exact Multinomial Tests}
\author{Johannes Resin\thanks{This work has been supported by the Klaus Tschira Foundation.
  		The author \add{would like} to thank Tilmann Gneiting, Alexander I.\ Jordan and Sebastian Lerch for helpful comments, discussions and continued encouragement \add{as well as two anonymous reviewers for their constructive comments}.}\hspace{.2cm}\\
    	Heidelberg Institute for Theoretical Studies%, Heidelberg, Germany 
    	\\ Karlsruhe Institute of Technology%, Karlsruhe, Germany
	}
\maketitle

\begin{abstract}
	This work proposes a new method for computing acceptance regions of exact multinomial tests. From this an algorithm is derived, which finds exact $p$-values for tests of simple multinomial hypotheses. Using concepts from discrete convex analysis, the method is proven to be exact for various popular test statistics, including Pearson's chi-square and the log-likelihood ratio. The proposed algorithm improves greatly on the naive approach using full enumeration of the sample space. However, its use is limited to multinomial distributions with a small number of categories, as the runtime grows exponentially in the number of possible outcomes.

	The method is applied in a simulation study, and uses of multinomial tests in forecast evaluation are outlined. Additionally, properties of a test statistic using probability ordering, referred to as the ``exact multinomial test'' by some authors, are investigated and discussed. The algorithm is implemented in the accompanying R package \texttt{ExactMultinom}.
	
\end{abstract}

\section{Introduction}

Multinomial goodness-of-fit tests feature prominently in the statistical literature and a wide range of applications. Tests relying on asymptotics have been available for a long time and have been rigorously studied all through the 20\textsuperscript{th} century. The use of various test statistics has been investigated with Pearson's chi-square and the log-likelihood ratio statistic being vital examples. These statistics are members of the general family of power divergence statistics \citep{CR84}. With the widespread availability of computing power, Monte Carlo simulations and exact methods have also gained popularity.

\citet{TH73} and \citet{KG80} used the ``exact multinomial test'', which orders samples by probability\add{,} to assess the accuracy of asymptotic tests \add{of a simple null hypothesis against an unspecified alternative}. In the words of \citet{CR89}, this ``has provided much confusion and contention in the literature''. In accordance with \citet{GP75} and \citet{RA75}, they conclude that the asymptotic fit of a test should be assessed using the appropriate exact test based on the test statistic in question. Nevertheless, the exact multinomial test is intuitively appealing, and, as \citet{KG80} put it, ``[i]n the absence of [...] a specific alternative, it is reasonable to assume that outcomes with smaller probabilities under the null hypothesis offer a stronger evidence for its rejection and should belong to the critical region''. In Section \ref{Sec:2}, an asymptotic chi-square approximation to the exact multinomial test is derived, and an exemplary comparison of popular test statistics in terms of power is provided.

Regardless of the test statistic used, \add{computing} an exact $p$-value by fully enumerating the sample space is computationally challenging, as the test statistic and the probability mass function have to be evaluated at every possible sample of which there are $\binom{n+m-1}{m-1} = \mathcal{O}(n^{m-1})$ for samples of size $n$ with $m$ categories. An improvement on this method has been proposed by \citet{BFT04} \add{for the family of power divergence statistics. Other} approaches \add{aimed at exact Pearson's chi-square and log-likelihood ratio tests} exist \add{\citep[see for example][]{BOP92,Hir97,Rah03,KN06}}. In this work, a new approach to exact multinomial tests is investigated.

\begin{figure}\centering
	\includegraphics[scale = 0.5]{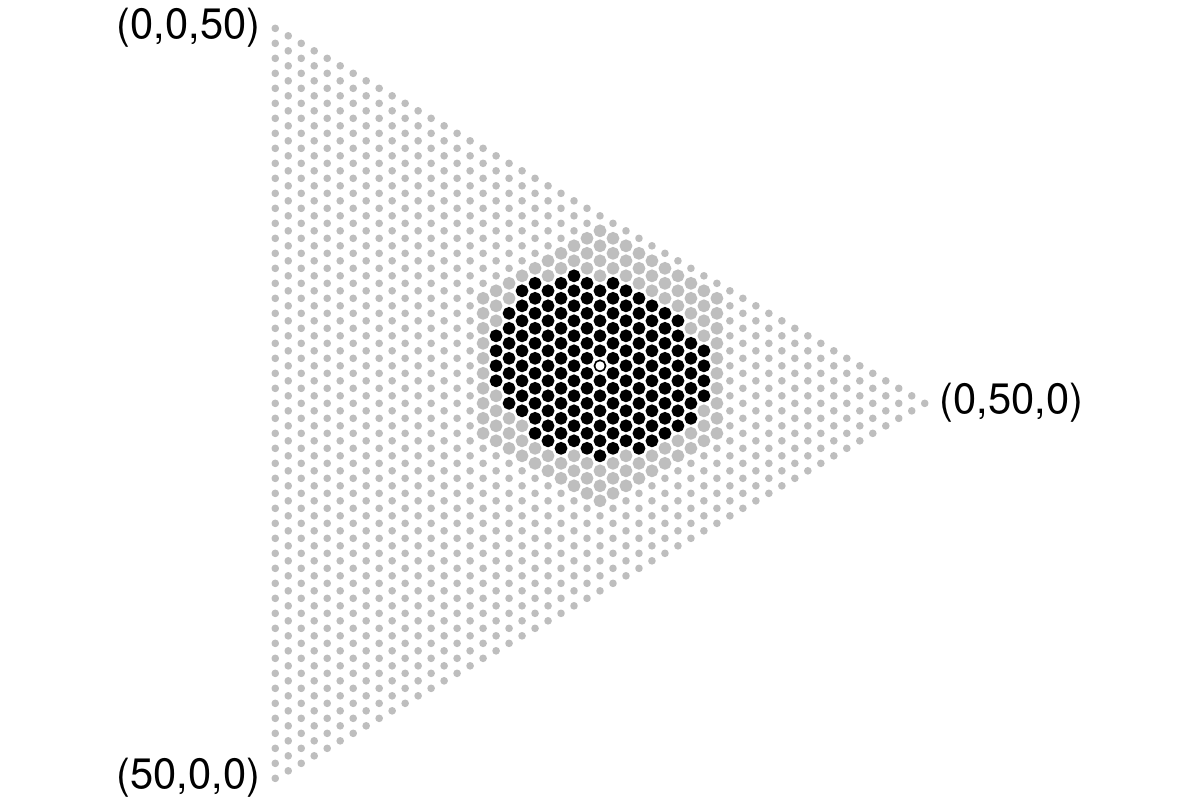}
	\caption{An acceptance region (\add{black dots}) at level $\alpha = 0.05$ for the null $\pi = (\frac{2}{10},\frac{5}{10},\frac{3}{10})$ and samples of size $n = 50$ with $m = 3$ categories. Only \add{points} within the ball (\add{big dots}) around the expectation (\add{hollow dot}) have to be considered to find this \add{region}.
	}
	\label{Fig:AReg}
\end{figure}

The key observation underlying the proposed algorithm is that acceptance regions at arbitrary levels contain relatively few points, which are located in a neighborhood of the expected value under the null hypothesis as illustrated in Figure \ref{Fig:AReg}\add{, and a}n acceptance region can be found by iteratively evaluating points within a ball of increasing radius around the expected value (w.r.t.\ the Manhattan distance). \add{The algorithm utilizes this by computing an exact $p$-value} from the probability mass of the smallest acceptance region that does not contain \add{the} observation. If $p$-values below an arbitrary threshold are not \add{computed} exactly, the runtime of the algorithm is guaranteed to be asymptotically faster than the approach using full enumeration as the diameter of any acceptance region essentially grows at a rate proportional to the square root of the sample size. This is detailed and proven to work for various popular test statistics in Section \ref{Sec:3}.

Furthermore, the algorithm is illustrated to work well in applications detailed in Section \ref{Sec:4}. In particular, the algorithm's runtime is compared to the full enumeration method in a simulation study, and the resulting $p$-values are used to assess the fit of asymptotic chi-square approximations and investigate differences between several test statistics. \add{As an application in forecast evaluation, the use of multinomial tests for uncertainty quantification}
%to quantify the gravity of discrepancies in forecast probabilities and outcome frequencies 
within the so-called calibration simplex \citep{Wil13} is outlined and justified. 

The R programming language \citep{R} has been used for all \add{computations} throughout this work. An implementation of the proposed method is provided within the R package \texttt{ExactMultinom} \add{\citep{EM}}.
%available at the CRAN package repository \add{(\url{https://CRAN.R-project.org/package=ExactMultinom})}.

\section{A Brief Review on Testing a Simple Multinomial Hypothesis}\label{Sec:2}

Consider a multinomial experiment $X = (X_1,\dots,X_m)$ summarizing $n \in \N$ i.i.d.\ trials with $m \in \N$ possible outcomes. Let 
$$\Deltam := \{p\in[0,1]^m \mid p_1 + \ldots + p_m = 1\}$$
denote the \emph{unit $(m-1)$-simplex} or \emph{probability simplex} and 
$$\Deltamn = \{x \in \N_0^m \mid x_1 + \ldots + x_m = n\}$$
the \add{sample space, which is a} \emph{regular discrete $(m-1)$-simplex}. The distribution of $X$ is characterized by a parameter $p = (p_1,\dots,p_m) \in \Deltam$ encoding the occurrence probabilities of the outcomes on any trial, or $X \sim \mathcal{M}_m(n,p)$ for short. The multinomial distribution $\mathcal{M}_m(n,p)$ is fully described by the probability mass function (pmf)
$$f_{n,p} \colon \Deltamn \rightarrow [0,1],x\mapsto n! \prod_{j = 1}^m \frac{p_j^{x_j}}{x_j!}.$$
%\frac{n!}{x_1!x_2!\dots x_m!}\cdot p_1^{x_1}\cdot\ldots\cdot p_m^{x_m}.$$

Suppose that the true parameter $p$ is unknown. Consider the simple null hypothesis $p = \pi$ for some $\pi \in \Deltam$. The agreement of a realization $x\in\Deltamn$ of $X$ with the null hypothesis is typically quantified by means of a test statistic $T\colon \Deltamn\times\Deltam \rightarrow \R$. Given such a test statistic $T$ and presuming from now on that w.l.o.g.\ high values of $T(x,\pi)$ indicate `extreme' observations under the null distribution $\P_\pi$, the $p$-value of $x$ is defined as the probability
\begin{equation}
p_T(x,\pi) := \P_\pi(T(X,\pi) \geq T(x,\pi)) %= \sum_{\substack{y \in \Deltamn \\ T(y,\pi) \geq T(x,\pi)}} f_{n,\pi}(y)
\label{Eq:Pval}
\end{equation}
of observing an observation that is at least as extreme under the null hypothesis.

The \emph{family of power divergence statistics} introduced by \citet{CR84} offers a variety of test statistics for multinomial goodness-of-fit tests. It is defined as
\begin{equation}
T^\lambda(x,\pi) := \frac{2}{\lambda(\lambda + 1)} \sum_{j = 1}^m x_j \left(\left(\frac{x_j}{n\pi_j}\right)^\lambda - 1\right) \text{ for } \lambda \in\R\setminus\{-1,0\}
\label{Eq:PowDivFam}
\end{equation}
and as the pointwise limit in (\ref{Eq:PowDivFam}) for $\lambda \in \{-1,0\}$. Notably, this includes \emph{Pearson's chi-square} statistic
$$T^{\chi^2}(x,\pi) := \sum_{j = 1}^m \frac{(x_j - n\pi_j)^2}{n\pi_j} = \sum_{j = 1}^m \frac{x_j^2}{n\pi_j} - n = T^1(x,\pi)$$
as well as the \emph{log-likelihood ratio} (or $G$-test) statistic
$$T^{G}(x,\pi) := 2\log \frac{f_{n,\frac xn}(x)}{f_{n,\pi}(x)} = 2 \sum_{j = 1}^m x_j \log \frac{x_j}{n\pi_j} = T^0(x,\pi).$$
Under a null hypothesis with $\pi_i > 0$ for all $i = 1,\dots,m$, every power divergence statistic is asymptotically chi-square distributed with $m-1$ degrees of freedom.

A natural test statistic arises if an `extreme' observation is simply understood to mean an unlikely one, that is, if the pmf itself is used as test statistic. In what follows, a strictly decreasing transformation of the pmf is used instead, which ensures that large values of the test statistic indicate extreme observations. Furthermore, this strictly decreasing transformation is chosen such that the resulting test statistic is asymptotically chi-square distributed.
To this end, let $\Gamma$ denote the Gamma function and
$$\bar{f}_{n,p} \colon \{x \in \R_{\geq 0}^m \mid x_1 + \ldots + x_m = n\} \rightarrow \R,x\mapsto \Gamma(n+1) \prod_{j = 1}^m \frac{p_j^{x_j}}{\Gamma(x_j+1)}$$
the continuous extension of the pmf $f_{n,p}$ to the convex hull of the discrete simplex $\Deltamn$\add{. T}he \emph{probability mass test statistic} \add{is defined} as
$$T^\P(x,\pi) := -2\log \frac{f_{n,\pi}(x)}{\bar{f}_{n,\pi}(n\pi)}.$$
Obviously, the choice of strictly decreasing transformation does not affect the (exact) $p$-value given by (\ref{Eq:Pval}) for $T = T^\P$.
The following theorem gives rise to an asymptotic approximation of $p$-values derived from the probability mass test statistic, which has not been studied previously. In the simulation study of Section \ref{Sec:4}, the fit of this approximation is assessed empirically using exact $p$-values \add{computed} with the new method for samples of size $n = 100$ with $m = 5$ categories.

\begin{Thm}
	If $X \sim \mathcal{M}_m(n,\pi)$ follows a multinomial distribution with $n \in \N$ and $\pi \in \Deltam$ such that $\pi_j > 0$ for $j = 1,\dots,m$, then $T^\P(X,\pi)$ converges in distribution to a chi-square distribution $\chi^2_{m-1}$ with $m-1$ degrees of freedom as $n \rightarrow \infty$.
	\label{Thm:AsApprox}
\end{Thm}

\begin{proof}
	By Lemma \ref{Lemma:DiffLLRProb} (in Appendix \ref{App:A}), the difference between the log-likelihood ratio and the probability mass statistic is
	$$T^\P(X,\pi) - T^{G}(X,\pi) = \sum_{j = 1}^m \left(\log \frac{X_j}{n\pi_j} + \mathcal{O}(1/X_j) - \mathcal{O}(1/n)\right).% =  \sum_{j = 1}^m \log \frac{X_j}{n\pi_j} + \mathcal{O}_p(1/n)
	$$
	Clearly, the bounded terms converge to zero in probability\add{,} and the $\log \frac{X_j}{n\pi_j}$ terms converge to zero in probability by the continuous mapping theorem. Hence, the probability mass statistic has the same asymptotic distribution as the log-likelihood ratio statistic.
\end{proof}

In what follows, the focus is on the chi-square, log-likelihood ratio and probability mass statistics.

\subsection{Acceptance \add{r}egions}

As outlined in the introduction, acceptance regions are of major importance to the idea pursued in this work. Given a test statistic $T$, the \emph{acceptance region at level} $\alpha > 0$ is defined \add{using $p$-values given by \eqref{Eq:Pval}} as
$$A^T_{n,\pi}(\alpha) := \{x \in \Deltamn \mid \add{p_T(x,\pi) > \alpha}\}.$$
Equivalently, the acceptance region can be written as the \emph{sublevel set} of $T(\cdot,\pi)$ at \add{the lowest} $(1-\alpha)$-quantile $t_{1-\alpha}$ of $T(X,\pi)$ under the null hypothesis $X\sim\mathcal{M}_m(n,\pi)$, i.e.,
\begin{equation}
	A^T_{n,\pi}(\alpha) = \{x \in \Deltamn \mid T(x,\pi) \leq t_{1-\alpha}\}.
	\label{Eq:AR}
\end{equation}
%By construction, the probability mass test statistic often yields a smallest acceptance region, because it assigns the samples with largest probabilities to the acceptance region. This is the case precisely if \add{$\P_\pi(X \in A^{T^\P}_{n,\pi}(\alpha)) - \min_{x\in A_{n,\pi}^{T^\P}(\alpha)} \P_\pi(X = x) < (1-\alpha)$}.
\add{By construction, the probability mass test statistic assigns the samples with largest probabilities to the acceptance region. Therefore, it yields a smallest acceptance region precisely if removing any point from $A^{T^\P}_{n,\pi}(\alpha)$ yields a set with probability mass less than $1-\alpha$.}
If tests are randomized to ensure equal level and size of the test, this property can be refined to yield an optimality property of the probability mass test's critical function. Figure \ref{Fig:ARegs} illustrates acceptance regions for different test statistics.

In Section \ref{Sec:3}, it will be shown that acceptance regions of the chi-square, log-likelihood ratio and probability mass test statistic all grow at a rate $\mathcal{O}(n^\frac{m-1}{2})$, as their diameter grows at a rate $\mathcal{O}(\sqrt{n})$ if $\alpha > 0$ is fixed, see Proposition \ref{Prop:GrowthAR}.

\begin{figure}\centering
	\includegraphics[scale = 0.5]{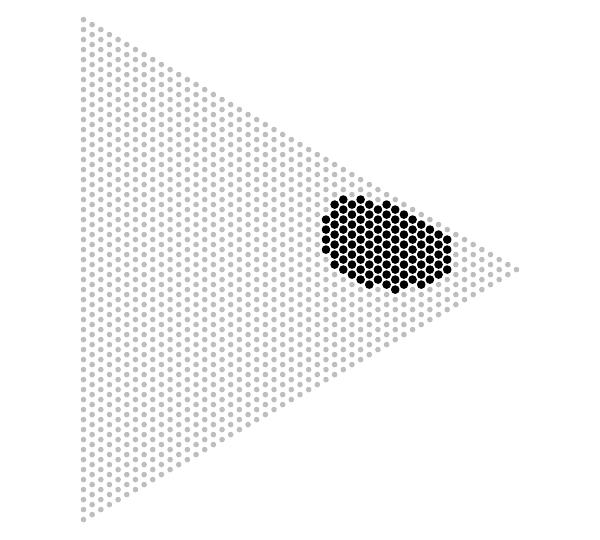}
	\includegraphics[scale = 0.5]{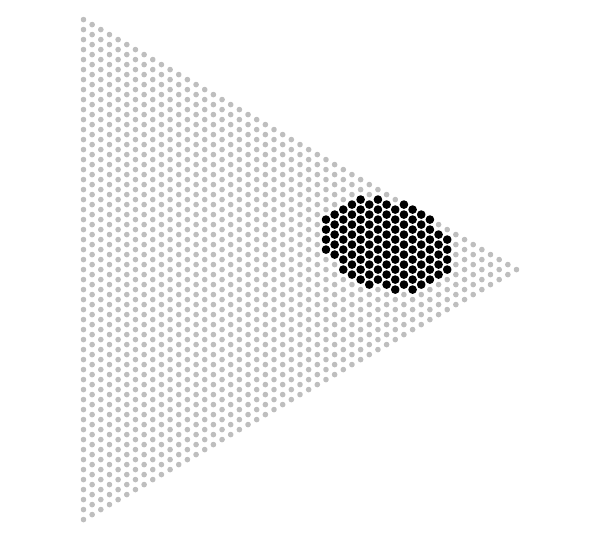}
	\includegraphics[scale = 0.5]{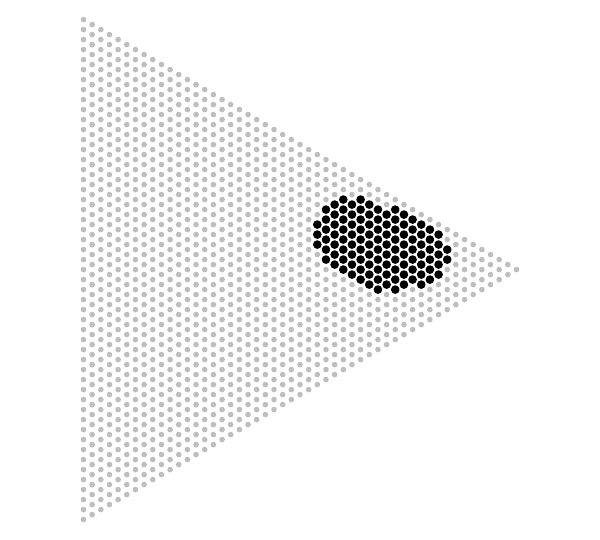}
	\caption{Acceptance regions \add{(black)} of probability mass (left), chi-square (center) and log-likelihood ratio (right) statistics at level $\alpha = 0.05$ for $n = 50$ and $\pi = (\frac{1}{10},\frac{7}{10},\frac{2}{10})$. The regions contain 108, 111 and 111 points, respectively (left to right). The tests are of size $0.0495, 0.0492$ and $0.0481$, respectively.}
	\label{Fig:ARegs}
\end{figure}

\subsection{Power and \add{b}ias}

The \emph{power function} of a test $T$ of the null hypothesis $p = \pi$ at level $\alpha$ is
$$\Deltam \rightarrow [0,1], p \mapsto 1-\P_p(T(X) \in A_{n,\pi}^T(\alpha)),$$
which is the probability of rejecting the null hypothesis at level $\alpha$ if the true parameter is $p$.  The \emph{size} of a test is its power at $p = \pi$. A test $T$ is said to be \emph{unbiased} (for the null $p = \pi$ at level $\alpha$) if its power is minimized at $p = \pi$.

In \add{the} case of the uniform null hypothesis, i.e., $\pi = (\frac{1}{m},\dots,\frac{1}{m})$, \citet[Theorem 2.1]{CS75} proved that the power function increases away from $p = \pi$ for \add{test} statistics of the form
%$$\text{Reject } p = \pi \quad \Leftrightarrow \quad T(x) = \sum_{j = 1}^m h(x_j) > t$$
$$T(x) = \sum_{j = 1}^m h(x_j)$$
if $h$ is a convex function. They concluded that tests based on the chi-square and the log-likelihood ratio test statistic are unbiased for the uniform null hypothesis. %(\cite{CS75}, Remark 2.1). 
As a corollary to their theorem, it shall be noted that this also applies to the probability mass test statistic.

\begin{Cor}[to \citealp{CS75}, Theorem 2.1]
	The probability mass test is unbiased for the uniform null hypothesis $p = \pi = (\frac{1}{m},\dots,\frac{1}{m})$.
\end{Cor}

\begin{proof}
	Since the probability mass statistic can be written as
	$$T^\P(x,\pi) = 2\sum_{j=1}^m \log\Gamma(x_j+1) - x_j\log\pi_j - \log\frac{\Gamma(n\pi_j+1)}{\pi_j^{n\pi_j}},$$
	this is an immediate consequence of the fact that the Gamma function is logarithmically convex on the positive real numbers, which is part of a characterization given by the Bohr-Mollerup theorem \citep[Theorem 2.4.2]{BW10}.
\end{proof}

Many authors \citep[e.g.,][]{WK72,CR84,WOK87,PP03} have conducted small sample studies to investigate the power of chi-square, log-likelihood ratio and other tests. \add{When conducting such} studies, $\pi,n$ and $\alpha$ need to be chosen, all of which influence the resulting power function. Furthermore, it is frequently infeasible to assess the power function across all alternatives, and so alternatives of interest need to be picked.
Therefore, most of these studies focused on the case of the uniform null hypothesis. In this case, the chi-square test has greater power for alternatives that assign a large proportion of the probability mass to relatively few \add{categories}, whereas the log-likelihood ratio test has greater power for alternatives that assign considerable probability mass to many \add{categories} \citep[see also][]{KL80}.

\begin{figure}\centering
	\includegraphics[scale = 0.5]{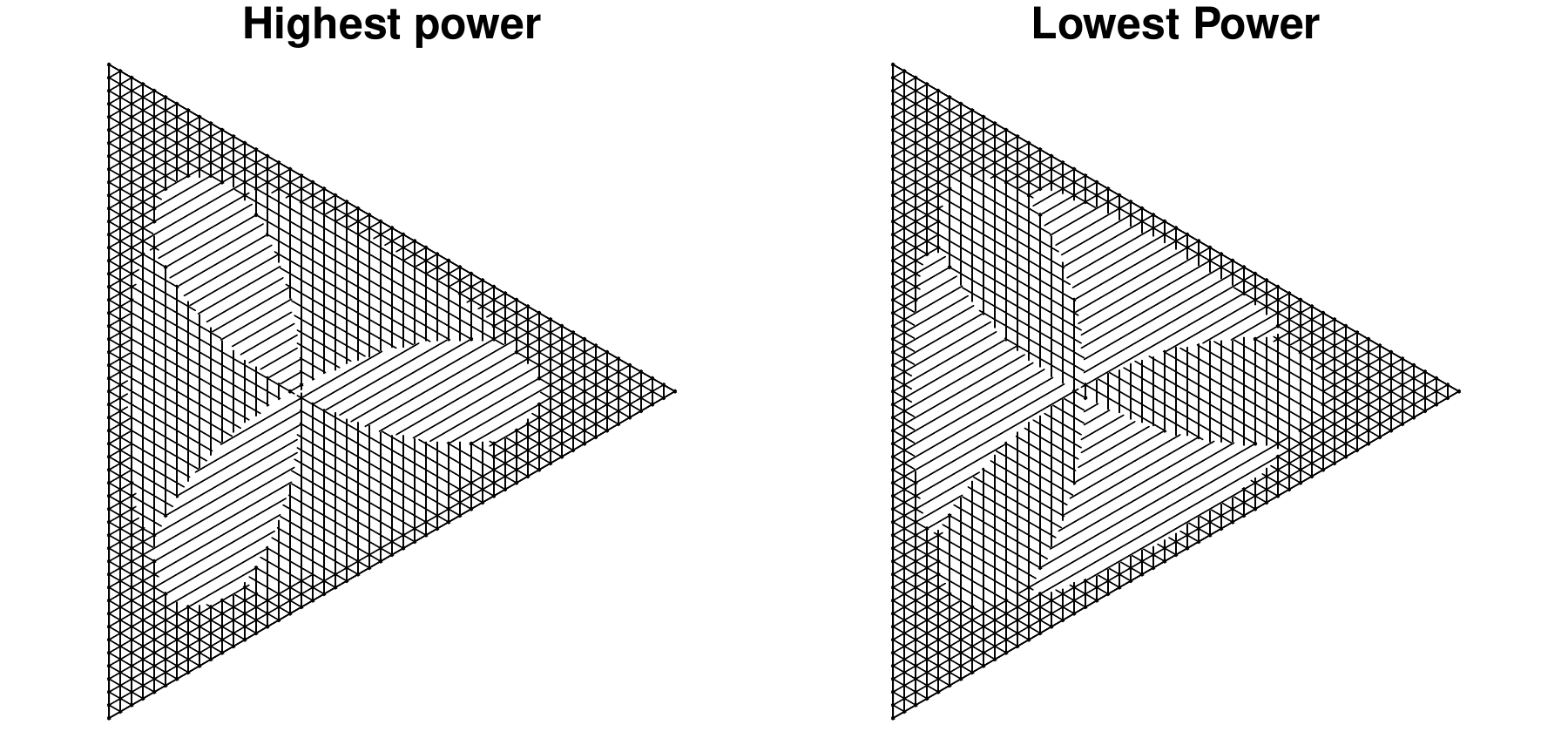}
	\includegraphics[scale = 0.5]{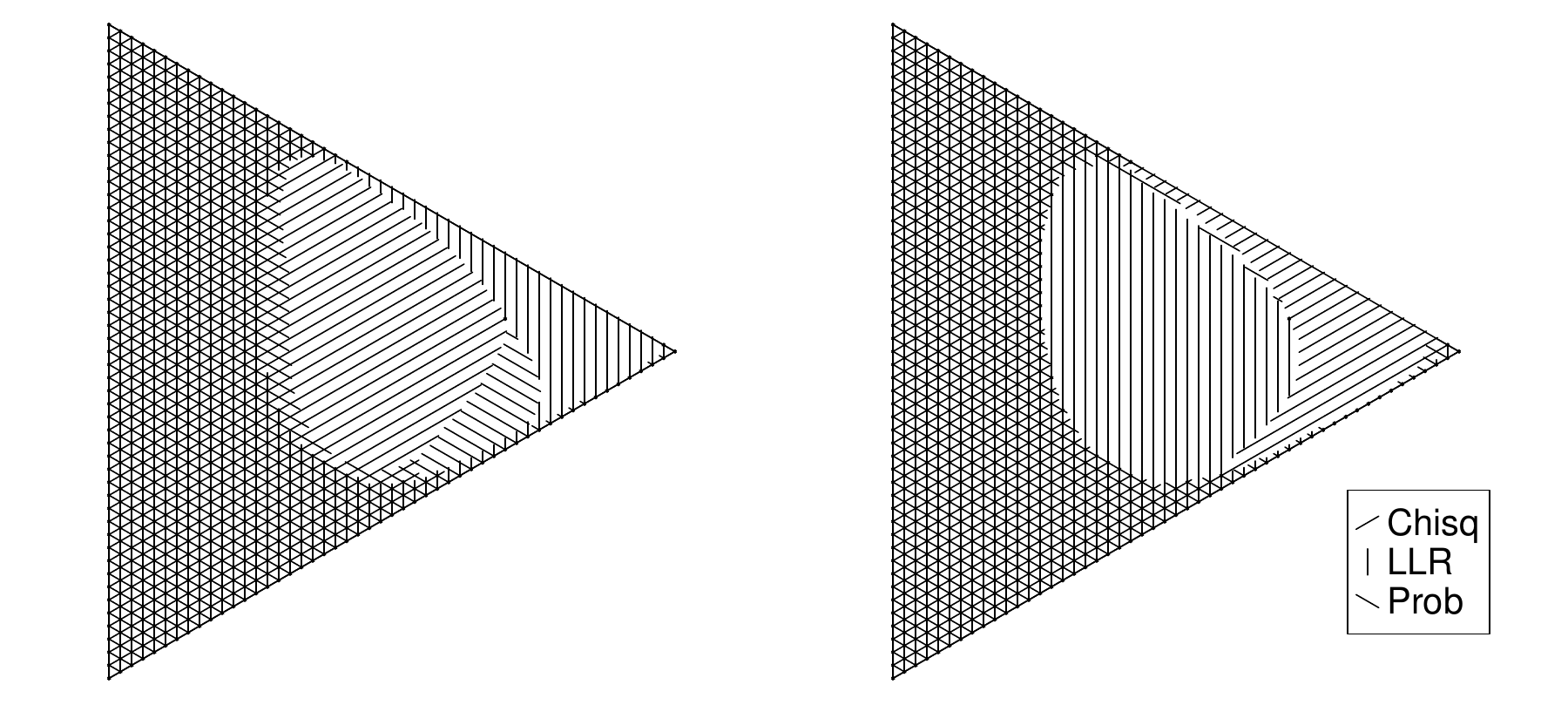}
	\caption{Ternary plots indicating \add{which randomized tests of size $\alpha = 0.05$ yields} the \add{highest (left) and lowest (right) power} for the uniform null hypothesis $\pi = (\frac{1}{3},\frac{1}{3},\frac{1}{3})$ (top) and $\pi = (\frac{1}{10},\frac{7}{10},\frac{2}{10})$ (bottom) for $n = 50$. \add{Overlapping lines} indicate nearly equal powers (\add{$\text{difference} < 10^{-5}$}).}
	\label{Fig:PowBestWorst}
\end{figure}

In the ternary case, that is, if $m = 3$, comparisons on the full probability simplex are visually accessible. Figure \ref{Fig:PowBestWorst} illustrates, which of the three test statistics yields the highest and lowest power across the full ternary probability simplex.
\add{As the actual test size, which is frequently smaller than the level $\alpha$, depends on the test statistic, the resulting power functions are difficult to compare directly.} To account for this, the tests are randomized to ensure that their respective size matches the level. For a test $T$ and level $\alpha$, let $s_{n,\pi}(T,\alpha) = 1-\P_\pi(T(X) \in A^T_{n,\pi}(\alpha))$ denote the actual size of the test. The \emph{critical function}
$$\phi \colon \Deltamn \rightarrow [0,1], x \mapsto \begin{cases}
0,& \text{if } T(x,\pi) < t_{1-\alpha}, \\
\frac{\alpha - s_{n,\pi}(T,\alpha)}{\P_\pi(T(X) = t_{1-\alpha})},& \text{if } T(x,\pi) = t_{1-\alpha}, \\
1,& \text{if } T(x,\pi) > t_{1-\alpha},
\end{cases}$$ defines a randomized test\footnote{Randomized tests like this traditionally arise in the theory of uniformly most powerful tests, see for example \citet[Chapter 3]{LR05}.}, which rejects the null hypothesis with probability $\phi(x)$ if $x$ is observed. The power function of the randomized version of a test $T$ at level $\alpha$ is
$$p \mapsto \sum_{x \in \Deltamn} \phi(x) \P_p(X = x) = 1 - \sum_{x \in A_{n,\pi}^T(\alpha)} (1-\phi(x)) \P_p(X=x).$$
With this, the probability mass test minimizes the acceptance region in the sense that it minimizes the sum
$$\sum_{x \in \Deltamn} (1-\phi(x))$$
across all randomized tests $\phi$ with $\sum_x \phi(x) f_{n,\pi}(x) = \alpha$.

\begin{figure}\centering
	\includegraphics[scale = 0.5,page=1]{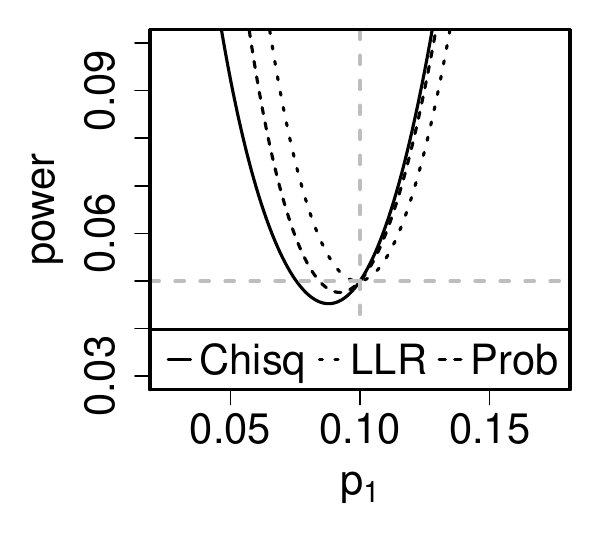}
	\includegraphics[scale = 0.5,page=2]{figs/power_plots_ex.pdf}
	\includegraphics[scale = 0.5,page=3]{figs/power_plots_ex.pdf}
	\caption{Power functions \add{of randomized tests of size $\alpha = 0.05$} along alternatives given by $p(p_i,i), i = 1,2,3$ \add{with} null hypothesis $\pi = (\frac{1}{10},\frac{7}{10},\frac{2}{10})$ and sample size $n = 50$.}
	\label{Fig:PowPlots}
\end{figure}

Figure \ref{Fig:PowBestWorst} \add{shows} that the probability mass test and the log-likelihood ratio test for the uniform null hypothesis \add{at level $\alpha = 0.05$} are the same for \add{$n = 50$}. \add{T}his is \add{a coincidence, and} for other choices of $\alpha$ (e.g., $\alpha = 0.13$, for which coincidentally the probability mass statistic yields the same acceptance region as the chi-square statistic) the acceptance regions differ, and so do the power functions. 

Figure \ref{Fig:PowPlots} quantitatively compares powers along alternatives of the form
$$p(q,i) = (\tilde{q}\pi_1,\dots,\tilde{q}\pi_{i-1},q,\tilde{q}\pi_{i+1},\dots,\tilde{q}\pi_m) \in \Deltam\text{ with } \tilde{q} = \frac{1-q}{1-\pi_i}$$
for $i = 1,\dots,m$ and $q \in [0,1]$. This yields parametrizations of the lines through $\pi$ and a corner of the probability simplex. The figures illustrate that in the case $n = 50, \pi = (\frac{1}{10},\frac{7}{10},\frac{2}{10})$ and $\alpha = 0.05$, the log-likelihood ratio test, arguably, does not show any visible bias, whereas the chi-square test shows the most bias. The power function of the probability mass test lies in between the other power functions across most of the probability simplex, and so the probability mass test might serve as a good compromise in terms of power.

\section{Exact $p$-Values via Acceptance Regions}\label{Sec:3}

Throughout this section, $T$ is \add{a} test statistic\add{,} and $m,n \in \N$ and $\pi \in \Deltam$ are \add{fixed}. To ease notation, the subscripts in the pmf of the null distribution are omitted, i.e., $f = f_{n,\pi}$ and the test statistic $T$ is considered as a function on the sample space only, i.e., $T(\cdot) = T(\cdot,\pi)$. Let
$$d\colon \R^m\times\R^m \rightarrow \R_{\geq 0}, (x,y) \mapsto \frac{1}{2}\Vert x - y \Vert_1 = \frac{1}{2}\sum_j \vert x_j - y_j \vert$$
be a rescaled version of the \emph{Manhattan distance} and 
$$B_r(y) = \{x\in\Deltamn \mid d(x,y)\leq r\}$$
the \add{discrete} ball with radius $r \in \N$ and center $y\in \Deltamn$. Furthermore, $e_i = (\delta_{ij})_{j = 1}^m$ denotes the $i$-th vector of the standard basis of $\R^m$\add{, where $\delta_{ij}$ is the Kronecker delta}.

\subsection{Finding acceptance regions using discrete convex analysis}

As alluded to in the introduction, an acceptance region $A = A^T_{n,\pi}(\alpha)$ for $\alpha \in (0,1)$ can be found without enumerating all points of the sample space $\Deltamn$, but only considering points in some ball around the expected value for many test statistics. Specifically, if $T$ is \emph{weakly quasi M-convex}, that is, \add{if} for all distinct $x,y \in \Deltamn$ there exist indices $i,j \in \{1,\dots,m\}$ such that $x_i > y_i, x_j < y_j$ and
$$T(x-e_i+e_j) \leq T(x) \quad\text{ or }\quad T(y+e_i-e_j) \leq T(y),$$
the following theorem\add{, which is proven at the end of this subsection,} holds.

\begin{Thm}
%	Let $T$ be weakly quasi M-convex. Let $y \in \Deltamn$ and $r \in \N$ such that $\sum_{x \in B_r(y)} f(x) \geq 1-\alpha$ for some $\alpha \in (0,1)$. 
%	If there exists a subset $A \subset B_{r-1}(y)$ of the form $A = \{x \in B_r(y)\mid T(x)\leq t\}$ such that $\sum_{x\in A} f(x) \geq 1-\alpha$, then the smallest such subset is the acceptance region $A^T_{n,\pi}(\alpha)$.
\add{Let $T$ be weakly quasi M-convex, and suppose $y \in \Deltamn$, $r \in \N$ and $\alpha \in (0,1)$ are such that $%\P_\pi(X \in B_r(y))
\sum_{x \in B_r(y)} f(x)
\geq 1-\alpha$. Let $t \in \R$ be the smallest level such that the sublevel set $A = \{x \in B_r(y)\mid T(x)\leq t\}$ satisfies $\sum_{x\in A} f(x) \geq 1-\alpha$. %Let $A = \{x \in B_r(y)\mid T(x)\leq t\}$ for the smallest $t \in \R$ such that $\sum_{x\in A} f(x) \geq 1-\alpha$. 
If $A \subseteq B_{r-1}(y)$, then $A$ is the acceptance region $A^T_{n,\pi}(\alpha)$.}
\label{Thm:MainObs}
\end{Thm}

Hence, an acceptance region can be found by iteratively enumerating a ball of increasing radius with arbitrary center until a sublevel set with enough probability mass is found and this sublevel set remains unchanged upon further increasing the ball\add{, as} illustrated in the introduction \add{for an acceptance region of the probability mass statistic}, see Figure \ref{Fig:AReg}.

The following proposition ensures that this approach can be applied to the chi-square, log-likelihood ratio and probability mass test statistics.

\begin{Prop}
	~
	\begin{enumerate}[label = \alph*)]
		\item The probability mass test statistic $T^\P$ is weakly quasi M-convex.
		\item The power divergence test statistic $T^\lambda$ is weakly quasi M-convex if $\lambda \geq 0$.
	\end{enumerate}
	\label{Prop:QMC}
\end{Prop}

\begin{proof}
	Throughout the proof, let $x,y \in \Deltamn$ such that $x\neq y$, and define the index sets 
	$$S^+ := \{i\mid x_i > y_i\} \quad \text{ and } \quad S^- := \{j\mid x_j < y_j\}.$$
	\begin{enumerate}[label = \alph*)]
		\item Let $T = T^\P$ and w.l.o.g.\ $T(x) \geq T(y)$. Then
		\begin{align*}
		T(y) - T(x) &= -2\log \frac{f(y)}{f(x)} = -2 \log\left(\prod_{i \in S^+} \frac{x_i!}{y_i!}\pi_i^{y_i - x_i} \cdot \prod_{j \in S^-} \frac{x_j!}{y_j!}\pi_j^{y_j - x_j}\right) \\
		&= -2\log\left(\prod_{i \in S^+} \prod_{k = 1}^{x_i - y_i} \frac{y_i + k}{\pi_i} \cdot \prod_{j \in S^-} \prod_{k = 1}^{y_j - x_j} \frac{\pi_j}{x_j + k}\right) \leq 0.
		\end{align*}
		%		\begin{align*}
		%		f(y) &= f(x) \cdot \prod_{i \in S^+} \frac{x_i!}{y_i!}p_i^{y_i - x_i} \cdot \prod_{j \in S^-} \frac{x_j!}{y_j!}p_j^{y_j - x_j} \\
		%		&= f(x) \cdot \prod_{i \in S^+} \prod_{k = 1}^{x_i - y_i} \frac{y_i + k}{p_i} \cdot \prod_{j \in S^-} \prod_{k = 1}^{y_j - x_j} \frac{p_j}{x_j + k},
		%		\end{align*}
		Both double products contain an equal number of multiplicands (since $\sum_j x_j = \sum_j y_j = n$) and are nonempty (since $x \neq y$). As the entire product is at least 1, there exist indices $i \in S^+$ and $j\in S^-$ and natural numbers $k^+ \leq x_i - y_i$ and $k^-\leq y_j-x_j$ such that the second inequality holds in
		$$\frac{\pi_j}{x_j + 1} \geq \frac{\pi_j}{x_j + k^-} \geq \frac{\pi_i}{y_i + k^+} \geq \frac{\pi_i}{x_i}.$$
		Therefore, the inequality
		$$T(x - e_i + e_j) = T(x) - 2\log\left(\frac{x_i}{\pi_i}\cdot\frac{\pi_j}{x_j + 1}\right) \leq T(x)$$ holds.
		%$$f(x-e_i + e_j) = f(x)\cdot\frac{x_i}{p_i}\cdot\frac{p_j}{x_j + 1} \geq f(x)$$
		\item See Appendix \ref{App:B}.
	\end{enumerate}
	\vspace{-\baselineskip}
\end{proof}

The rest of this section is devoted to the proof of Theorem \ref{Thm:MainObs}. For further details on weak quasi M-convexity and discrete convex analysis in general, see \citet{Mur03}. 

Weakly quasi M-convex functions have the important property that their sublevel sets are weakly quasi M-convex sets \citep[Theorem 3.10]{MS03}. A subset $M \subset \Deltamn$ is \emph{weakly quasi M-convex} if for all distinct $x,y \in M$ there exist indices $i,j \in \{1,\dots,m\}$ such that $x_i > y_i, x_j < y_j$ and
$$x - e_i + e_j \in M \quad\text{ or }\quad y + e_i - e_j \in M.$$
Equivalently, this can be characterized as follows.

\begin{Lemma}
	A subset \add{$M \subset \Deltamn$} is weakly quasi M-convex if and only if 
	for all $x,y \in M$ and $d = d(x,y)$ there exists a sequence $x_0,x_1,\dots,x_d \in M$ with $x_0 = x, x_d = y$ and $d(x_i,x_{i+1}) = 1$ for all $i = 0,1,\dots,d-1$.
	\label{Lemma:NewChar}
\end{Lemma}

\begin{proof}~
	\begin{itemize}
		\item[``$\Rightarrow$'':] By induction on $d$: Let $x,y \in M$ and $d = d(x,y)$.
		If $d = 0$, then $x = x_0 = y$ satisfies the condition.
		If $d > 0$, \add{there exist $i,j$ such that $x_i > y_i, x_j < y_j$ and $x_{d-1} = y + e_i - e_j \in M$ (or $x_{d-1} = x - e_i + e_j \in M$, in which case interchanging $x$ and $y$ and $i$ and $j$ yields the former formula for $x_{d-1}$) by weakly quasi M-convexity of $M$.}
%		define $x_{d-1} = y + e_i - e_j$ for some $i,j$ such that $x_i > y_i, x_j < y_j$ and w.l.o.g.\ $x_{d-1} \in M$. 
		Then $d(x_{d-1},y) = 1$ and
		\begin{align*}			
		d(x,x_{d-1}) &= \frac{1}{2}\Bigg(\sum_{k\neq i,j} \vert x_k - y_k \vert + \underbrace{\vert x_i - (y_i + 1)\vert}_{=\vert x_i - y_i\vert - 1} + \underbrace{\vert x_j - (y_j - 1) \vert}_{=\vert x_j - y_j \vert - 1}\Bigg) \\
		&= \frac{1}{2}(\Vert x - y \Vert_1 - 2) = d-1.
		\end{align*}
		By induction hypothesis, there exists a sequence $x_0,x_1,\dots,x_{d-1} \in M$ such that, $x = x_0,x_1,\dots,x_{d-1},x_d = y \in M$ is the sought-after sequence.
		\item[``$\Leftarrow$'':] Let $x,y \in M\add{, x\neq y}$, and $d = d(x,y)$. Let $x_0,x_1,\dots,x_d$ be a sequence as in the lemma. As $d(x,x_1) = 1$, there exist $i,j$ such that $x_1 = x - e_i + e_j$. %(since the assumptions $\suppp(x-x_1) = \emptyset$ or $\suppm(x-x_1) = \emptyset$ yield $\sum_i x_i \neq \sum_i x_{1i}$, contradicting (i)). 
		Furthermore, $x_i > y_i$ and $x_j < y_j$, since
		\begin{align*}
		d-1 &= \sum_{l = 1}^{d-1} d(x_l,x_{l+1}) \\
		&\geq d(x_1,y) = \frac12 \Bigg(\sum_{k \neq i,j} \vert x_k - y_k\vert + \vert x_i-1 - y_i\vert + \vert x_j+1 - y_j\vert\Bigg)
		\end{align*}
		\add{yields a contradiction otherwise.}
	\end{itemize}
	\vspace{-\baselineskip}
\end{proof}

With this, the theorem can be proven as follows.

\begin{proof}[Proof of Theorem \ref{Thm:MainObs}]
	Let \add{$t\in\R$} be minimal such that $A = \{x \in B_r(y) \mid T(x) \leq t\}$ has probability mass $\sum_{x \in A} f(x) \geq 1-\alpha$ and \add{$A \subseteq B_{r-1}(y)$, i.e., $A \cap (B_r(y)\setminus B_{r-1}(y)) = \emptyset$}. Furthermore, fix $a \in A$ such that $T(a) = t$. \add{Recall that the acceptance region $A_{n,\pi}^T(\alpha)$ is the sublevel set \eqref{Eq:AR} at $t_{1-\alpha}$, and note that $t_{1-\alpha} \leq t$ holds as $\P_\pi(T(X) \leq t) \geq \sum_{x \in A} f(x) \geq 1-\alpha$.}
	
	Assume there exists some $b \in A_{n,\pi}^T(\alpha)\setminus A$, i.e., \add{$T(b) \leq t_{1-\alpha} \leq t$. Then $b \notin B_r(y)$ by construction of $A$}. \add{Since} the test statistic $T$ is weakly quasi M-convex\add{,} the sublevel set \add{$L = \{x \in \Deltamn \mid T(x) \leq t\} \supseteq A$} is weakly quasi M-convex. By Lemma \ref{Lemma:NewChar}, there exists a sequence $a = a_0,a_1,\dots,a_d = b \in L$ with $d = d(a,b)$ and $d(a_i,a_{i+1}) = 1$ for all $i = 0,1,\dots,d-1$. By the triangle inequality $d(a_i,y) - 1 \leq d(a_{i+1},y) \leq d(a_i,y) + 1$. Thus, there exists some $j \in \{1,\dots,d-1\}$ such that \add{$d(a_j,y) = r$}, a contradiction (as $T(a_j) \leq t$ \add{but} $a_j \in B_r(y)\setminus B_{r-1}(y)$). Therefore, $A_{n,\pi}^T(\alpha) \subseteq A$, and hence $A = A_{n,\pi}^T(\alpha)$, \add{because $t$ is minimal}.
\end{proof}

\subsection{Calculating a $p$-value}

As described in the previous subsection, an acceptance region can be determined by \add{taking an} arbitrary point and increasing the radius of a ball around this \add{center} point until the acceptance region is found using the criterion provided by Theorem \ref{Thm:MainObs}. Obviously, \add{the center of the ball should lie within the acceptance region,
%a point that is not within the acceptance region is not a practical starting point, and 
ideally at its} center, to minimize the necessary iterations and number of points for which to evaluate the pmf and the test statistic. The expected value $\E X = n\cdot p$ of the multinomial distribution, which is the center of mass of all probability weighted points in the discrete simplex, is known, and \add{it is} close to the center of mass of the acceptance region, as the \add{region} contains most of the mass. Therefore, a \add{point} close to the expected value \add{is a suitable center for the ball}.

The $p$-value of an observation $x$ can be found by \add{computing} the total probability of the largest acceptance region not containing the observation. \add{However}, this region can be large if the $p$-value of the observation is very small. To avoid this, Algorithm \ref{Algo:P-val} does not \add{compute} very small $p$-values precisely, but only determines precise $p$-values above a certain threshold $\theta$ and otherwise states that the $p$-value is smaller than the threshold $\theta$. Figure \ref{Fig:Algo} \add{shows} the points evaluated by Algorithm \ref{Algo:P-val} for \add{an observation} with $p$-value greater, respectively smaller than some threshold $\theta$.

%\begin{algorithm}[t]
%	\begin{algorithmic}
%		\REQUIRE Observation $x \in \Deltamn$, hypothesis $\pi \in \Deltam$, threshold $0 < \theta \ll 1$
%		\ENSURE Exact $p$-value $p \in [\theta,1]$ or $0$ if the $p$-value is less than $\theta$.
%		%\STATE $f = f_{n,\pi}$
%		\STATE Calculate $y \in \Deltamn$ minimizing $d(y,\E_\pi X)$
%		\IF{$T(x) \leq T(y)$} \STATE Set $y = x$ 
%		\ENDIF
%		\STATE Initialize $r = 1$, $\text{SumProb} =0$
%		\REPEAT
%		\STATE Add $f(z)$ to SumProb for points \add{$z \in B_r(y)\setminus B_{r-1}(y)$} with $T(z) < T(x)$
%		\STATE Increment $r = r+1$
%		\UNTIL{$T(x) \leq \min \{T(z)\mid d(y,z) = r\}$ \OR SumProb $> 1-\theta$}
%		\IF{$\text{SumProb} \leq 1-\theta$} 
%		\RETURN $1- \text{SumProb}$
%		\ELSE 
%		\RETURN $0$ 
%		\ENDIF
%	\end{algorithmic}
%	\caption{Calculate exact $p$-value above some threshold.}
%	\label{Algo:P-val}
%\end{algorithm}

\begin{algorithm}[t]
	\DontPrintSemicolon
	\caption{\add{Compute} exact $p$-value above some threshold.}
	\label{Algo:P-val}
	\SetAlgoLined
	\KwIn{Observation $x \in \Deltamn$, hypothesis $\pi \in \Deltam$, threshold $0 < \theta \ll 1$}
	\KwOut{Exact $p$-value $p \in [\theta,1]$ or $0$ if the $p$-value is less than $\theta$}
	\add{compute} $y \in \Deltamn$ minimizing $d(y,\E_\pi X)$ \\
	\lIf{$T(x) \leq T(y)$}{set $y = x$} 
	initialize \add{$r = 0$}, $\text{SumProb} = 0$ \\
	\Repeat{$T(x) \leq \min \{T(z)\mid d(y,z) = r\}$ or \textnormal{SumProb} $> 1-\theta$}{
	add $f(z)$ to SumProb for points \add{$z \in B_r(y)\setminus B_{r-1}(y)$} with $T(z) < T(x)$ \\
	increment $r = r+1$}
	\lIf{$\textnormal{SumProb} \leq 1-\theta$}{\Return $1- \text{SumProb}$}
	\lElse{\Return $0$}
\end{algorithm}

\begin{figure}[t]\centering
	\includegraphics[scale = 0.5]{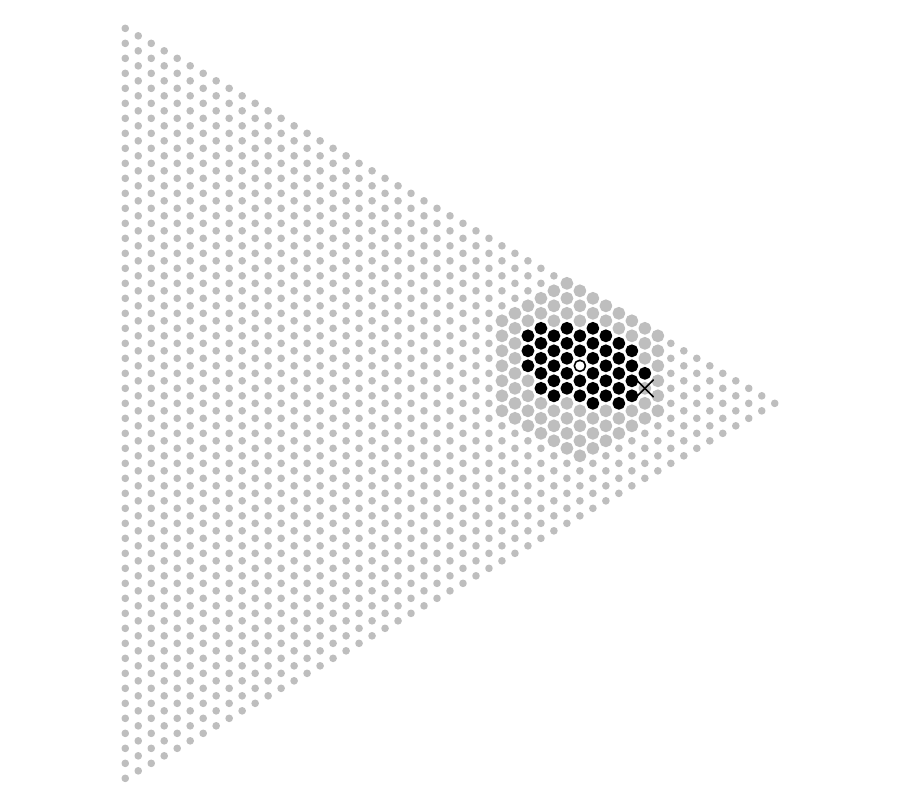}
	\includegraphics[scale = 0.5]{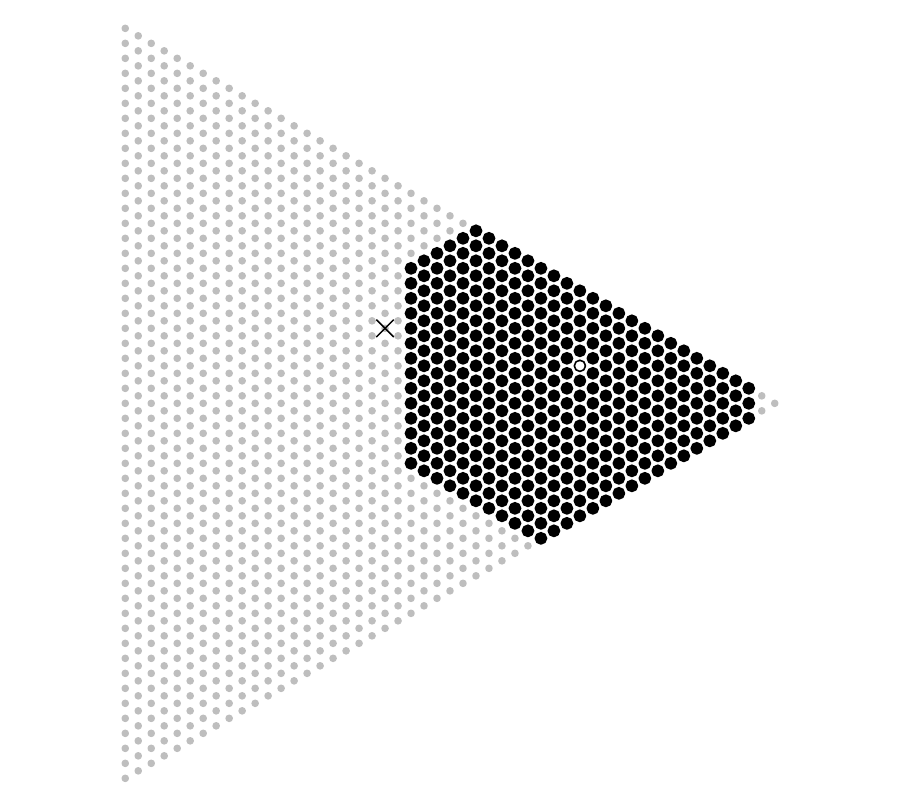}
	\caption{\add{Points (big dots)} in \add{$\Omega_{3,50}$} for which the probability mass and test statistic are evaluated given the marked observations $x = (4,40,6)$ (left) and $x = (10,20,20)$ (right) under the null hypothesis $\add{\pi} = (\frac{1}{10},\frac{7}{10},\frac{2}{10})$ and $T = T^\P$. The $p$-values are 0.3049 %0.3048903 
	(left) and less than $\theta = 0.0001$ (right). The \add{black} region on the left \add{is} the smallest acceptance region not containing the \add{observation $x$}. }
	\label{Fig:Algo}
\end{figure}

\subsection{Implementation}
\label{Sec:Implementation}

Enumeration of the full sample space can be implemented using a simple recursion. A similar, more complicated recursive scheme can be employed to enumerate the samples at a given radius $r$ in the repeat-loop of Algorithm \ref{Algo:P-val}. This is implemented in the R package \texttt{ExactMultinom} using a C++ subroutine to allow for fast recursions.

As mentioned in the \add{introduction}, algorithms for \add{computing} exact multinomial tests superior to the full enumeration method have been proposed in the literature. However, 
\add{none of these methods have been tailored to the probability mass test, and most of them do not produce ``strictly exact'' $p$-values \citep{KN06}. Appendix \ref{App:Comparison} provides a short overview of and comparison with other methods, which shows that Algorithm \ref{Algo:P-val} performs favorably in the setting of the simulation study of Section \ref{Sec:4}.}
%readily available open source implementations of these methods apparently do not exist. 
There are two packages implementing exact multinomial tests using full enumeration of the sample space in R, namely, \texttt{EMT} \citep{EMT} and \texttt{XNomial} \citep{XNomial}. Whereas \texttt{EMT} is written purely in R, the function \texttt{xmulti} of the \texttt{XNomial} package implements the full enumeration method using an efficient C++ subroutine for the recursion, which makes it a lot more efficient. %\add{For these reasons}, \texttt{xmulti} was selected as reference method.
\add{T}he implementation of Algorithm \ref{Algo:P-val} \add{simultaneously computes} $p$-values for the chi-square, log-likelihood ratio and probability mass test statistics, as \add{does} \texttt{xmulti}, and so comparability is ensured.

The current implementation of Algorithm \ref{Algo:P-val} accurately finds $p$-values of order roughly as small as $10^{-10}$. Smaller $p$-values will often lead to negative output because of limited computational precision in the addition of many floating point numbers. To ensure accurate results, I recommend to choose $\theta$ no less than $10^{-8}$ with the current implementation.

During early runs of the simulation study described in Section \ref{Sec:4}, it was noticed that the runtime of Algorithm \ref{Algo:P-val} tends to increase drastically if the null distribution contains \add{a} very small probability \add{$\pi_i \ll n^{-1}$ for some $i \leq m$}. 
%This is due to the acceptance regions becoming very flat and containing mostly points within a lower dimensional face of the discrete simplex for such null hypotheses. In this case, $n$ is too small for Proposition \ref{Prop:GrowthAR} below to take effect. 
\add{In this case, the acceptance region is very flat, containing mostly points within a lower dimensional face of the discrete simplex, as hits in category $i$ are improbable under the null. Hence, the asymptotic advantage of Algorithm \ref{Algo:P-val} discussed in the next subsection requires large sample size $n$ to take effect under sparse null hypotheses.}
As a heuristic, which turned out to be an effective remedy, the implementation does not enumerate entire balls if $n\cdot\pi_i < \frac12$, but only considers points $z\in\Deltamn$ with small $z_i$, by skipping all points $z$ for which $\P_\pi(X_i \geq z_i) < \theta\cdot 10^{-8}$.

\subsection{Runtime complexity}

\begin{figure}[t]\centering
	\includegraphics[scale = 0.5]{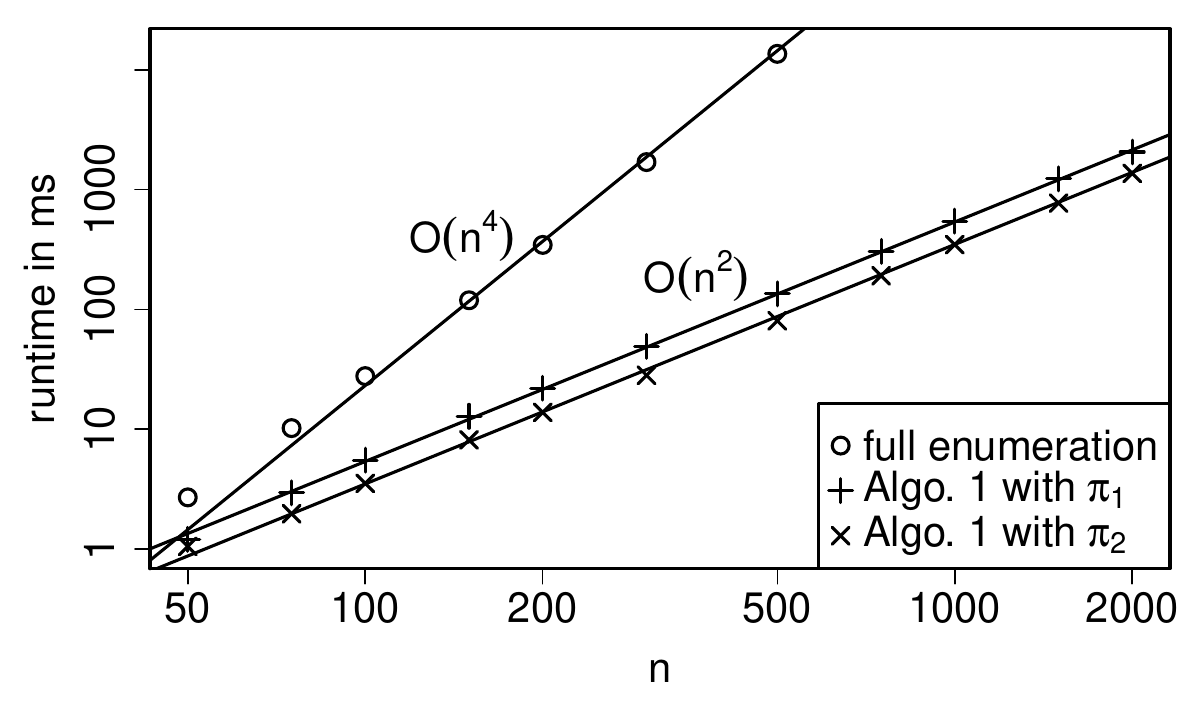}
	\caption{Runtime of the full enumeration method and Algorithm \ref{Algo:P-val} when enumerating a ball with probability mass $1-\theta$ for $\theta = 0.0001$ and null hypotheses $\pi_1 = (0.2,0.2,0.2,0.2,0.2)$ \add{and} $\pi_2 = (0.01,0.19,0.2,0.3,0.3)$. If the $p$-values of an observation are significantly larger than $\theta$, the runtime of Algorithm \ref{Algo:P-val} considerably decreases. Times are mean values from 10 runs.}
	\label{Fig:Time5}
\end{figure} 

The discrete simplex $\Deltamn$ contains $\vert\Deltamn\vert = \binom{n+m-1}{m-1}$ points, and so the full enumeration takes $\mathcal{O}(n^{m-1})$ operations to compute a $p$-value. In comparison, the acceptance regions at a fixed level $\alpha > 0$ only contain $\mathcal{O}(n^\frac{m-1}2)$ points, and this continues to hold for the smallest ball centered at the expected value containing the acceptance region, as proven by Proposition \ref{Prop:GrowthAR} below. Therefore, Algorithm \ref{Algo:P-val} only takes $\mathcal{O}(n^\frac{m-1}2)$ operations to determine a $p$-value above the threshold $\theta$. Figure \ref{Fig:Time5} shows runtime as a function of $n$ for $m = 5$. Whereas the runtime of the full enumeration method does \add{neither} depend on the choice of $\pi$ \add{nor on} the observation $x$, the runtime of Algorithm \ref{Algo:P-val} increases if the $p$-value of $x$ is small. Furthermore, the choice of $\pi$ also influences the runtime of \add{the implementation} with the uniform null hypothesis resulting in a longer runtime than sparse null hypotheses \add{(when applying the heuristic described at the end of Section \ref{Sec:Implementation})}. This is further investigated in the simulation study in Section \ref{Sec:4}. As the runtime increases exponentially in $m$, Algorithm \ref{Algo:P-val} is only feasible if the number of categories $m$ is small.

\begin{Prop} \add{Let $T \in \{T^{\chi^2},T^{G},T^\P\}$, $\alpha \in (0,1)$ and $\pi \in \Deltam$. Then there exists $c = c(\alpha,\pi)$} such that $A^T(\alpha) \subset B_{\sqrt{n}c}(n\pi)$ for sufficiently large $n$.
	\label{Prop:GrowthAR}
\end{Prop}

\begin{proof}
	Consider the canonical extension $\bar{T}$ of $T$ to $\barDeltamn = \{x \in \R_{\geq 0}^m \mid x_1+\ldots+x_m = n\}$ and let $B_r(n\pi) = \{x \in \barDeltamn\mid \frac12\Vert x-n\pi\Vert_1 \leq r\}$ a ball in $\barDeltamn$ with boundary $\partial B_r(n\pi) = \{x \in \barDeltamn\mid \frac12\Vert x-n\pi\Vert_1 = r\}$. 
	Let $r_0 = \min_j \pi_j > 0$ and $n_0 \in \N$. If $n \geq n_0$, then every $x \in \partial B_{\sqrt{nn_0}r_0}(n\pi)$ can be written as $x = x(n,x_0) := n\pi + \sqrt{nn_0}(x_0 - \pi)$ for some $x_0 \in \partial B_{r_0}(\pi)$. %namely $x_0 = ((x - n\pi)/\sqrt{n}/\sqrt{n_0} + \pi) n_0$
	
	Let $(t_{n,1-\alpha})$ be the sequence of \add{lowest} $(1-\alpha)$-quantiles of $T_n = T(X_n), X_n\sim \mathcal{M}_m(n,\pi)$ for $n \in \N$. As $T_n$ converges to $\chi^2_{m-1}$ in distribution, the sequence of quantiles converges to the $(1-\alpha)$-quantile $\chi^2_{m-1,1-\alpha}$ \citep[cf.][Lemma 21.2]{VV98}. Consequently, the maximum $t = \max_n t_{n,1-\alpha}$ exists, and the set $A_n = \{x \in \barDeltamn \mid \bar{T}(x) \leq t\}$ contains the acceptance region \add{$A_{n,\pi}^T(\alpha)$} for every $n$.
	
	As $\bar{T}$ is convex (\add{by} Lemma \ref{Lemma:Convex} in Appendix \ref{App:C}) and thus has convex sublevel sets, it suffices to show that $n_0$ can be chosen such that $\min_{x \in \partial B_{\sqrt{nn_0}r_0}(n\pi)} \bar{T}(x)$ converges to a value $>t$ to ensure that $\add{A_{n,\pi}^T(\alpha)} \subset A_n \subset B_{\sqrt{n}(\sqrt{n_0}r_0)}(n\pi)$ for sufficiently large $n$.
	
	In case $T = T^{\chi^2}$, observe that
	\begin{align*}
	\bar{T}(x(n,x_0)) &= 
	\sum_j \frac{(x_j(n,x_0) - n\pi_j)^2}{n\pi_j} %\\ &
	= \sum_j \frac{n_0(x_{0,j} - \pi_j)^2}{\pi_j}
	\end{align*}
	does not depend on $n$, and so the canonical extension $\bar{T}$ of the chi-square statistic at radius $\sqrt{nn_0}r_0$ is bounded from below by $b(n_0) = \min_{x \in \partial B_{r_0}(n_0\pi)} \bar{T}(x)$. This bound becomes arbitrarily large as $n_0$ is increased.
	
	In case $T = T^{G}$ or $T = T^\P$, if $n_0$ is fixed, $\bar{T}(x(n,x_0))$ converges uniformly to $\bar{T}^{\chi^2}(x(n,x_0))$ for $x_0 \in \partial B_{r_0}(\pi)$ (\add{by} Lemma \ref{Lemma:UnifConv} in Appendix \ref{App:C}). \add{Therefore,} $\min_{x \in \partial B_{\sqrt{nn_0}r_0}(n\pi)} \bar{T}(x)$ converges to $b(n_0)$.
\end{proof}

\section{Application}\label{Sec:4}

In this section, the use of the new method is illustrated in a simulation study. On the one hand, this serves to show the improvements in runtime in comparison to the full enumeration method. On the other hand, this sheds some light on the fit of the asymptotic approximation to the probability mass test provided by Theorem \ref{Thm:AsApprox} for a \add{moderate} sample size ($n = 100$).
As a practical application \add{in forecast evaluation}, the usage of exact multinomial tests to increase the information conveyed by the \emph{calibration simplex} \citep{Wil13}, a graphical tool used to assess ternary probability forecasts, is outlined.

\subsection{Simulation study}

For the simulation study, pairs $(\pi^{(1)},x^{(1)}),\dots,(\pi^{(N)},x^{(N)})$ of null hypothesis parameters and samples were generated as i.i.d.\ realizations of the random quantity $(P,X)$ with $P \sim \mathcal{U}(\Deltam)$ being uniformly distributed on the unit simplex and $X \mid P \sim \mathcal{M}_m(n,P)$. %Note: the uniform distribution on the unit simplex is the Dirichlet distribution $\text{Dir}(1)$ with concentration parameters equal to 1
\add{For} each pair, $p$-values were computed using various test statistics and algorithms. \add{Thereby}, no specific null hypothesis \add{had} to be chosen and instead a wide variety \add{was} considered. By drawing samples from the null hypotheses, $p$-values follow a uniform distribution on $[0,1]$. Various aspects of the tests and algorithms in question can be examined using the resulting rich data set and subsets thereof. 

The following results were obtained using $N = 10^6$ such pairs with samples of size $n = 100$ drawn from multinomial distributions with $m = 5$ \add{categories}. Exact $p$-values were computed using the implementation of Algorithm \ref{Algo:P-val} provided by the accompanying R package\add{.} To estimate the speedup achieved by the new method in this study, the full enumeration method provided by the \texttt{xmulti} function of the \texttt{XNomial} package \citep{XNomial} was applied to the first $10^4$ pairs. Essentially, the computational cost of the full enumeration is constant, independent of the null hypothesis at hand and the resulting $p$-value, whereas the cost of Algorithm \ref{Algo:P-val} increases as the $p$-value decreases and also varies with the null hypothesis.

\begin{figure}\centering
	\includegraphics[scale = 0.5]{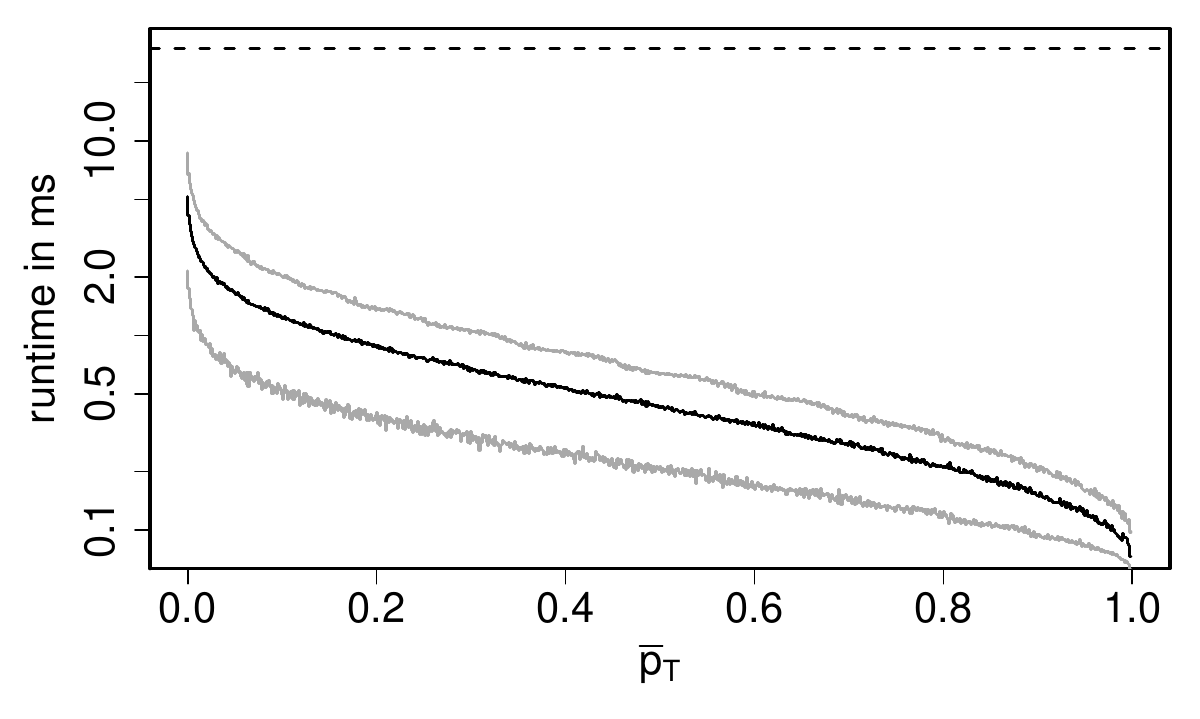}
	\caption{Runtime against mean $p$-value in groups of 1000 samples with similar mean $p$-value. The \add{black} line shows mean runtime per group, whereas the \add{grey} lines are the 5\% and 95\%-quantile. The \add{dashed} line shows the mean runtime using full enumeration.}
	\label{Fig:PValTime}
\end{figure}

The implementation of Algorithm \ref{Algo:P-val} took an average of 0.59 ms %0.5937891 
to \add{compute} a $p$-value, whereas the full enumeration took 29.76 ms %29.76492
on average, and so execution of the new method was about 50 times as fast. Perhaps surprisingly, Monte Carlo estimation (using \texttt{xmonte} from \texttt{XNomial}, which simulates 10000 samples by default) took almost twice as long (53.49 ms) %0.05349011
as the full enumeration. Figure \ref{Fig:PValTime} illustrates the connection between runtime and size of the resulting $p$-values for the new method. As there are other factors influencing the runtime and \add{the} implementation computes $p$-values for multiple statistics simultaneously, samples were ordered by their mean $p$-value $\bar{p}_T = \frac{1}{3}(p_{T^\P} + p_{T^{\chi^2}} + p_{T^{G}})$ and \add{put} in groups of 1000 samples \add{with} similar mean $p$-value (in particular, \add{the groups contain} samples \add{with $p$-values} in between the empirical \add{$(\frac{a}{1000})$- and $(\frac{a+1}{1000})$-quantile for $a = 0,\dots,999$}). The figure shows mean runtime in each group as well as the 5\%- and 95\%-quantile. 

\begin{figure}\centering
	\includegraphics[scale = 0.5]{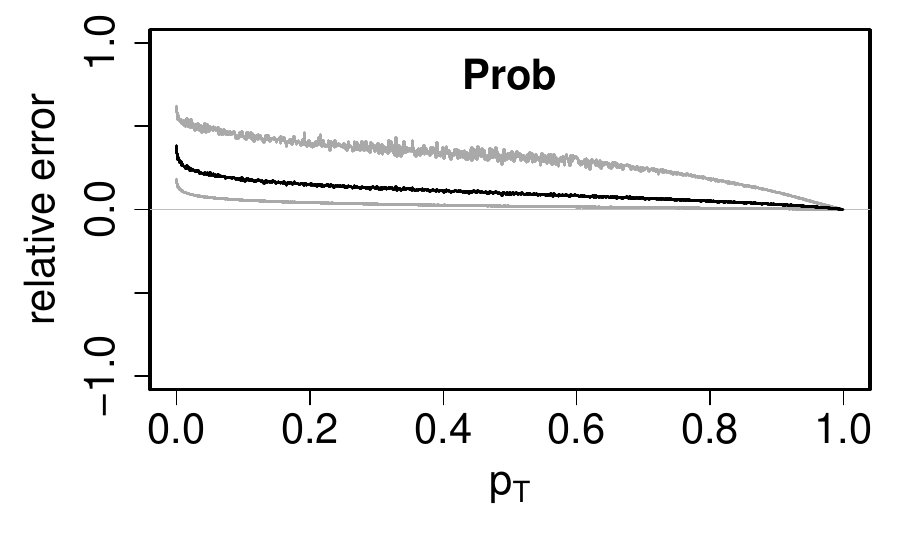}
	\includegraphics[scale = 0.5]{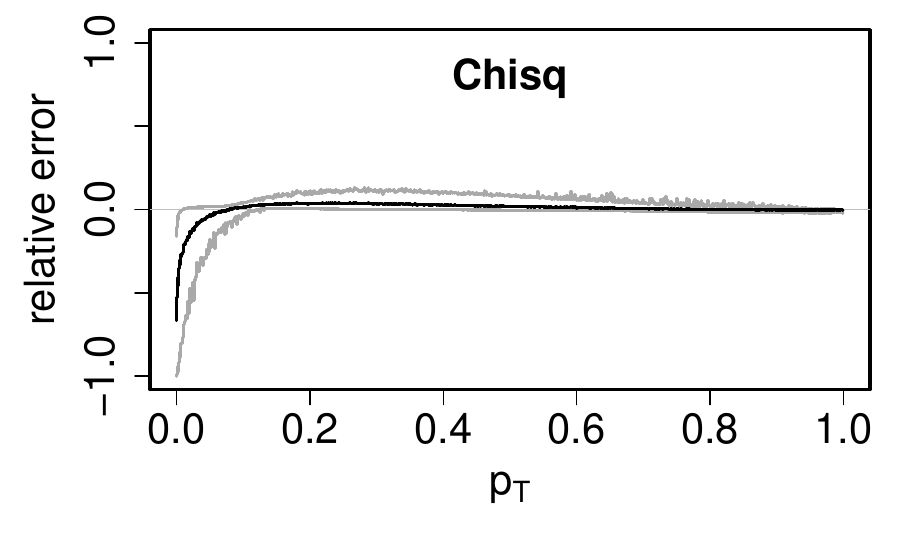}
	\includegraphics[scale = 0.5]{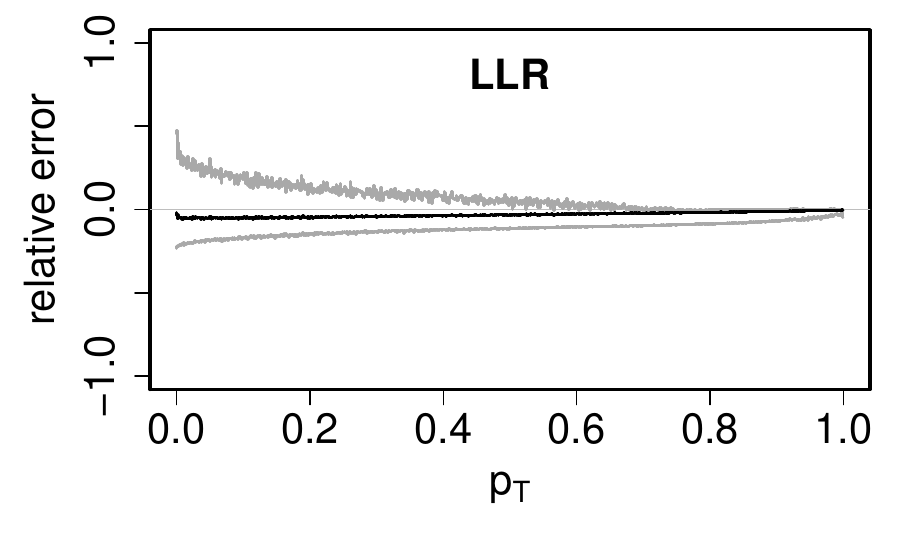}
	\caption{Relative errors of asymptotic approximation for probability mass (Prob), chi-square (Chisq) and log-likelihood ratio (LLR) test statistic. The plots were obtained using the same grouping scheme as in Figure \ref{Fig:PValTime}.}
	\label{Fig:AsError}
	
	\vspace*{\floatsep}
	
	\includegraphics[scale = 0.5]{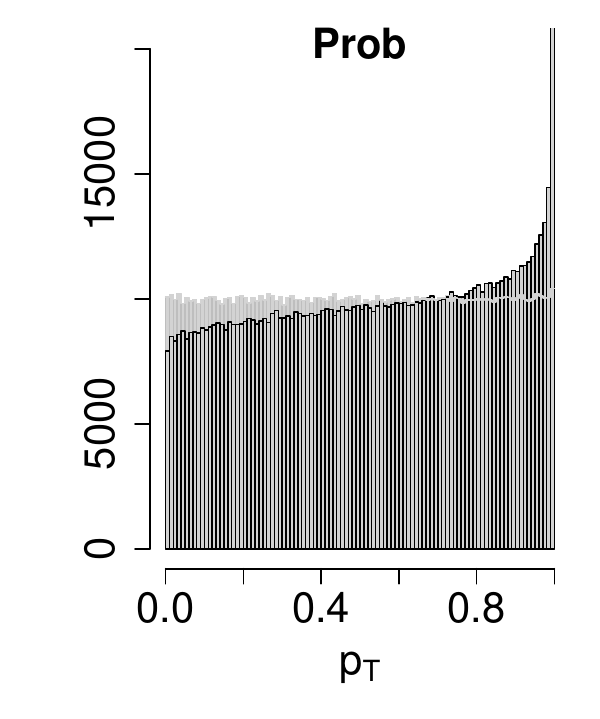}
	\includegraphics[scale = 0.5]{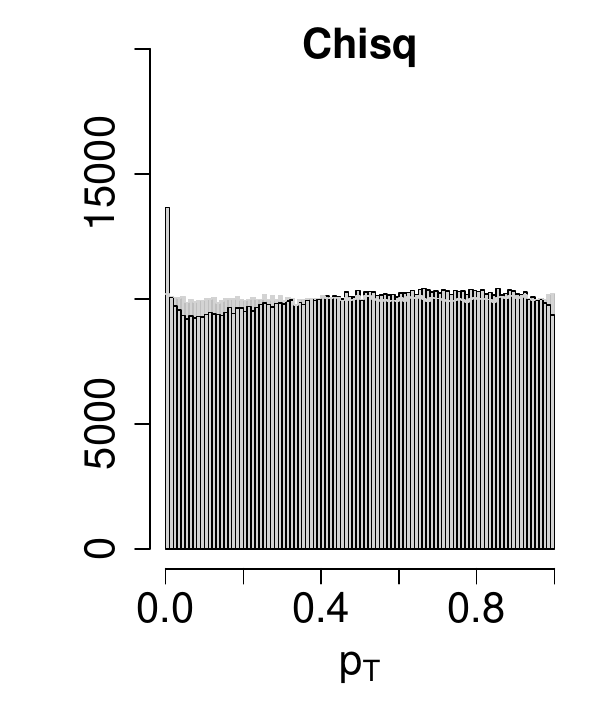}
	\includegraphics[scale = 0.5]{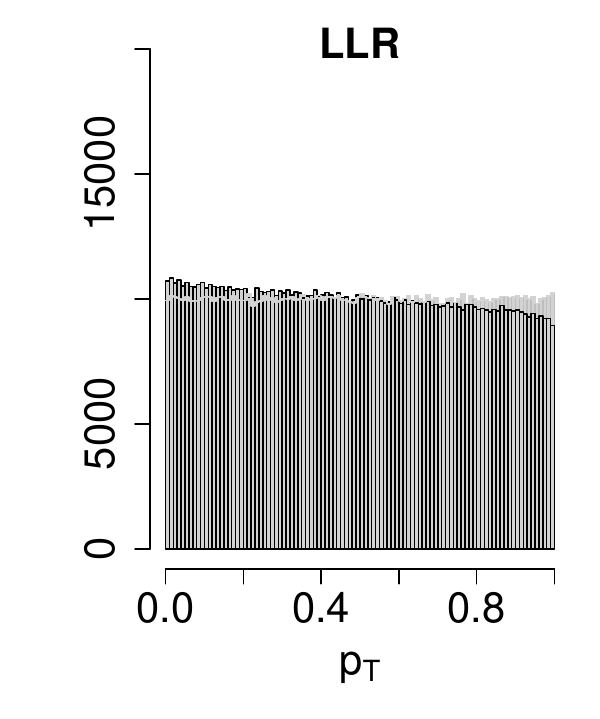}
	\caption{Histograms of asymptotic approximations to $p$-value\add{s} for probability mass (Prob), chi-square (Chisq) and log-likelihood ratio (LLR) test statistic in black. The \add{gray histograms show} respective exact $p$-values. The rightmost bar within the left histogram is not fully shown and extends further up to over 30000 counts.} %31272
	\label{Fig:HistPvals}
\end{figure}

\begin{figure}[t]\centering
	\includegraphics[scale = 0.5]{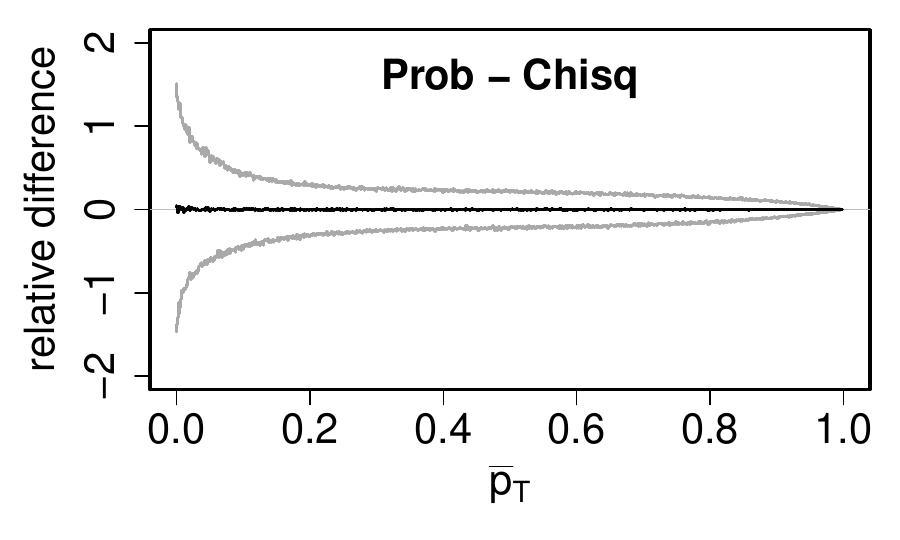}
	\includegraphics[scale = 0.5]{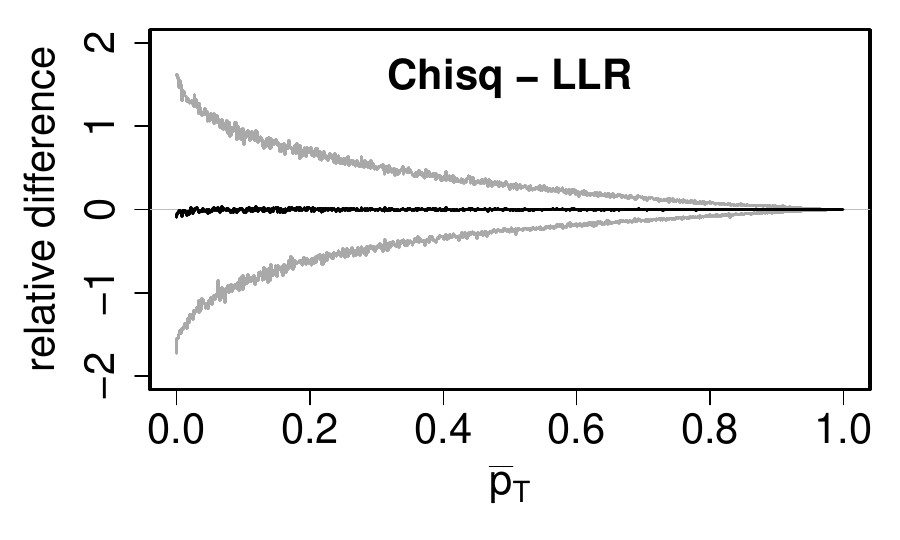}
	\includegraphics[scale = 0.5]{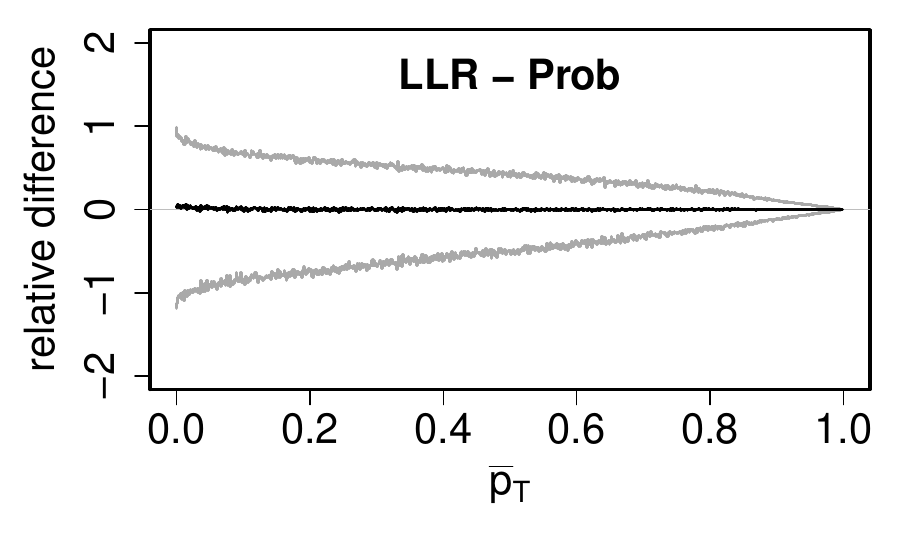}
	\caption{Relative differences between exact $p$-values of probability mass (Prob), chi-square (Chisq) and log-likelihood ratio (LLR) test statistic against mean of compared $p$-values. The plots were obtained using the same grouping scheme as in Figure \ref{Fig:PValTime}.}
	\label{Fig:RelDiff}
\end{figure}

To illustrate the fit of the classical chi-square approximation, the probability of a chi-square distribution with $m-1$ degrees of freedom exceeding the values of the test statistics for each pair were computed.
Figure \ref{Fig:AsError} shows relative errors of the asymptotic approximations to the $p$-values for the three test statistics. Given a test statistic $T$ and asymptotic approximation $\tilde p_T = \tilde p_T(x,\pi)$ to the exact $p$-value $p_T = p_T(x,\pi)$, the relative error \add{is the deviation from the exact value in parts of said value, $\frac{\tilde p_T - p_T}{p_T}$}. \add{T}he asymptotic approximation to the chi-square statistic is quite accurate in most cases, but tends to underestimate small $p$-values ($< 0.1$). The asymptotic approximation to the log-likelihood ratio statistic tends to slightly underestimate $p$-values on average. \add{While the exact $p$-values are \emph{valid} in that $\P_\pi(p_T(X,\pi) \leq \alpha) \leq \alpha$ for all $\alpha \in [0,1]$, underestimation may result in invalid $p$-values.} Asymptotic approximations of Pearson's chi-square and the log-likelihood ratio have been studied well, and the classical chi-square approximations can be improved by using moment corrections \citep[see][and references therein]{CR89}. Furthermore, the errors typically increase if some category has small expectation under the null hypothesis. \add{T}he approximation to the probability mass $p$-values provided by Theorem \ref{Thm:AsApprox} produces somewhat larger errors especially for large $p$-values, and \add{it} clearly overestimates the $p$-values. This is emphasized by the fact that within the simulation data only a vanishingly small number of $p$-values was slightly underestimated, all of which were well over 0.9. Figure \ref{Fig:HistPvals} illustrates how \add{estimation} errors influence the distribution of the resulting $p$-values. Whereas the exact $p$-values \add{appear to} follow a uniform distribution, \add{the} asymptotic $p$-values clearly deviate from uniformity. For the probability mass statistic, the asymptotic test \add{yields} a conservative test, whereas the asymptotic log-likelihood ratio test (and also the asymptotic chi-square test at small significance levels) is slightly anti-conservative.

\begin{table}%[t]
	\caption{Exact $p$-values $p_T$ and asymptotic $p$-values $\tilde p_T$ of five randomly selected pairs $(x,\pi)$ with $0.01 < p_{T^{G}}(x,\pi) < 0.1$.}
	\centering
	%\footnotesize
	\begin{tabular}{c|cc|cc|cc}
		$\pi$ & $p_{T^\P}$ & $\tilde p_{T^\P}$ & $p_{T^{\chi^2}}$ & $\tilde p_{T^{\chi^2}}$ & $p_{T^{G}}$ & $\tilde p_{T^{G}}$ \\ 
		\hline
		$(0.116 , 0.225 , 0.259 , 0.002 , 0.398)$ & 0.0068 & 0.0092 & 0.0190 & 0.0073 & 0.0126 & 0.0172 \\
		$(0.038 , 0.079 , 0.224 , 0.387 , 0.272)$ & 0.1150 & 0.1268 & 0.1437 & 0.1469 & 0.0361 & 0.0307 \\
		$(0.595 , 0.129 , 0.093 , 0.064 , 0.118)$ & 0.0447 & 0.0495 & 0.0477 & 0.0482 & 0.0719 & 0.0665 \\
		$(0.497 , 0.217 , 0.223 , 0.057 , 0.007)$ & 0.0761 & 0.0994 & 0.0803 & 0.0741 & 0.0461 & 0.0498 \\
		$(0.243 , 0.022 , 0.237 , 0.373 , 0.125)$ & 0.0474 & 0.0566 & 0.0508 & 0.0507 & 0.0628 & 0.0568 \\
	\end{tabular}
	\label{Tab:RandValues}
\end{table}

\add{F}igure \ref{Fig:RelDiff} shows relative differences between exact $p$-values obtained with the three test statistics. Given test statistics $T$ and $T'$, the relative difference between $p$-values $p_T = p_T(x,\pi)$ and $p_{T'} = p_T(x,\pi)$ \add{is} $\frac{p_T - p_{T'}}{\bar{p}_T}$\add{, where} $\bar{p}_T = \frac{p_T + p_{T'}}2$. It can be seen that the choice of test statistic can make quite a difference. A closer look at the simulation data revealed that these differences tend to be smaller if expectations for all categories are large under the null. To provide some numerical insights, Table \ref{Tab:RandValues} lists exact and asymptotic $p$-values.

\subsection{The calibration simplex}

Turning to an application in forecast verification, consider a random variable $X$ and a probabilistic forecast $F$ for $X$. For an introduction to probabilistic forecasting in general, see \citet{GK14}. A probabilistic forecast is said to be \emph{calibrated} if the conditional distribution of the quantity of interest given a forecast coincides with the forecast distribution, that is, 
\begin{equation}
X\mid F \sim F
\label{Eq:Cal}
\end{equation}
holds almost surely. Suppose now that $X$ maps to one of three distinct outcomes only. Then, a probabilistic forecast is fully described by the probabilities it assigns to each outcome. 
In this case, the calibration simplex \citep{Wil13} can be used to graphically identify discrepancies \add{between} predicted probabilities and conditional outcome frequencies. Given i.i.d.\ realizations 
$(f_1,x_1),\dots,(f_N,x_N)$
consisting of forecast probabilities (vectors within the unit 2-simplex) and observed outcomes encoded 1, 2 and 3, forecast-outcome pairs with similar forecast probabilities are grouped according to a tessellation of the probability simplex. Thereafter, calibration is assessed by comparing average forecast and actual outcome frequencies within each group.

\begin{figure}[t]\centering
	\includegraphics[scale = 0.5,trim = 0 0.65in 0 0.65in,clip]{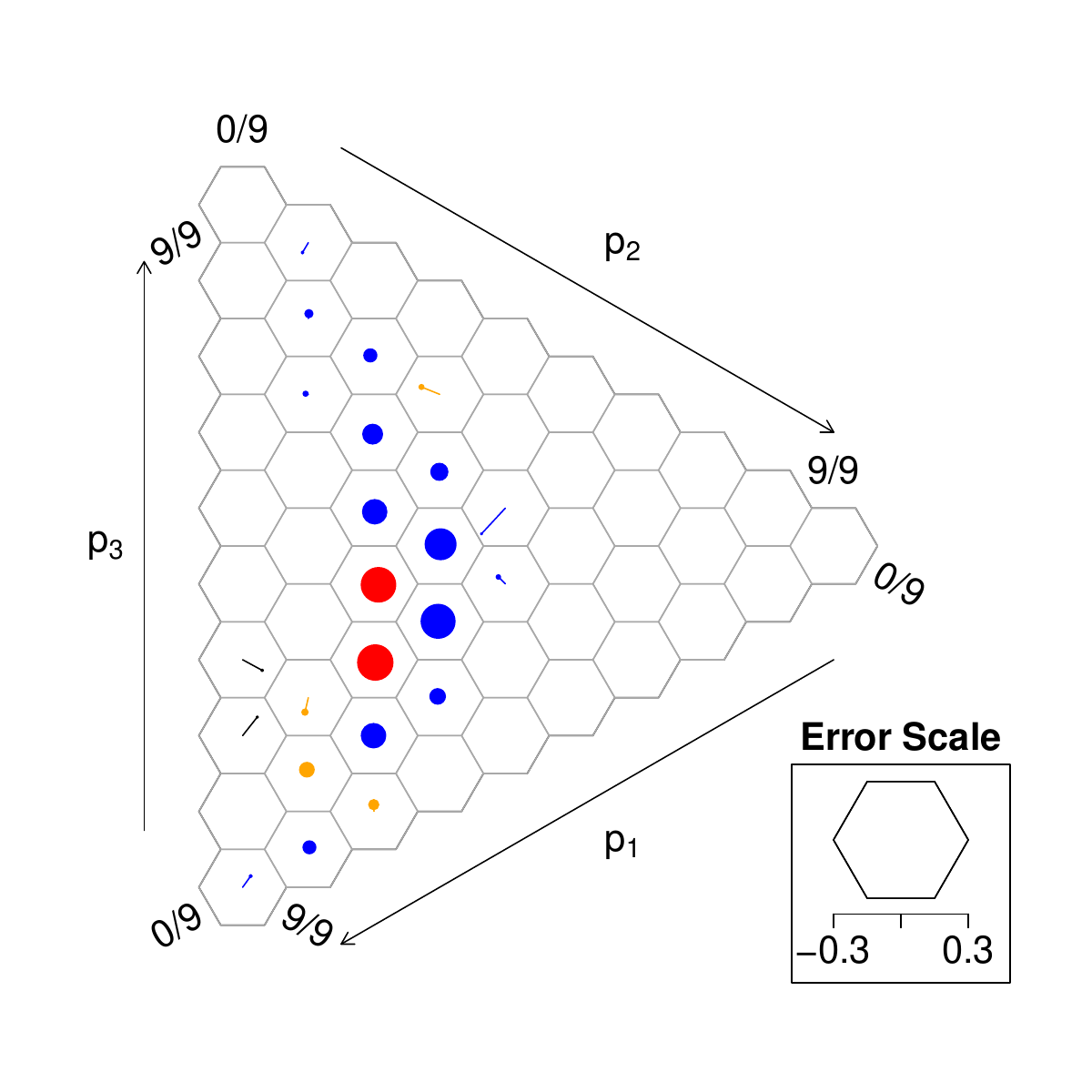}
	\caption{Calibration Simplex with color-coded $p$-values \add{from} the \add{log-likelihood ratio} statistic \add{evaluating a} total of \add{21,240} club soccer predictions by FiveThirtyEight (\SoccerPred) for matches from September 2016 until April 2019. Outcomes are encoded as $1 =$ ``home win'', $2 =$ ``draw'' and $3 =$ ``away win''. Only groups containing at least ten forecasts are shown.
	Blue indicates a $p$-value $p_{T^G}$ > 0.1, orange $0.1 > p_{T^G} \geq 0.01$, red $p_{T^G} < 0.01$ and black $p_{T^G} = 0$.
	}
	\label{Fig:CalSim}
\end{figure}

As illustrated in Figure \ref{Fig:CalSim}, the calibration simplex is a graphical tool \add{used} to conduct this comparison visually. The groups are determined by overlaying the probability simplex with a hexagonal grid. The circular dots correspond to nonempty groups of forecasts given by a hexagon. The dots' areas are proportional to the number of forecasts per group. A dot is shifted away from the center of the respective hexagon by a scaled version of the difference in average forecast probabilities and outcome frequencies. This provides valuable insight into the forecast's distribution and the conditional distribution of the quantity of interest. However, it is not apparent how big the differences may be merely by chance. 

If the forecast is calibrated, then, by (\ref{Eq:Cal}), the outcome frequencies $\bar{x}$ within a group of size $n$ with mean forecast $\bar{f}$ follow a generalized multinomial distribution (the multinomial analog of the Poisson binomial distribution), that is, a convolution of multinomial distributions $\mathcal{M}(1,f_i)$ with parameters $f_1,\dots,f_n \in \Deltam$. If these parameters only deviate little from their mean $\bar{f} = \frac1n\sum_i f_i$, then, presumably, the generalized multinomial distribution should not deviate much from a multinomial distribution with parameter $\bar{f}$. Under this presumption, multinomial tests can be applied to quantify the discrepancy within each group through a $p$-value. As the number of outcomes $m = 3$ is small, exact $p$-values are efficiently computed by Algorithm \ref{Algo:P-val} even for large sample sizes $n$. 

In Figure \ref{Fig:CalSim}, $p$-values obtained from the log-likelihood ratio statistic are conveyed through a coloring scheme. Note that a $p$-value will only ever be exactly zero, if an outcome is forecast to have zero probability and said outcome still realizes. Figure \ref{Fig:CalSim} was generated using the R package \texttt{CalSim} \citep{CalSim}.

The calibration simplex can be seen as a generalization of the popular reliability diagram. In light of this analogy, the use of multinomial tests to assess the statistical significance of differences in predicted probabilities and observed outcome frequencies serves the same purpose as consistency bars in reliability diagrams introduced by \citet{BS07}. Consistency bars are constructed using Monte Carlo simulation. To justify the above presumption, the multinomial $p$-values used to construct Figure \ref{Fig:CalSim} were compared to $p$-values \add{computed} from 10000 Monte Carlo samples obtained from the generalized multinomial distributions. To this end, the standard deviation of the Monte Carlo $p$-values was estimated using the estimated $p$-value in place of the true generalized multinomial $p$-value. Most of the multinomial $p$-values were quite close to the Monte Carlo estimates with an absolute difference less than two standard deviations, whereas two of them deviated on the order of 6 to 8 standard deviations from the Monte Carlo estimates, which nonetheless resulted in a relatively small absolute error. In particular, using the Monte Carlo estimated $p$-values did not change Figure \ref{Fig:CalSim}. As computation of the Monte Carlo estimates from the generalized multinomial distributions is computationally expensive, the multinomial $p$-values serve as a fast and adequate alternative. Further improving uncertainty quantification within the calibration simplex is a subject for future work.

\section{Concluding Remarks}\label{Sec:5}

A new method for \add{computing} exact $p$-values was investigated. It has been illustrated that the new method works well when the number $m$ of categories is small. This results in a concrete speedup in practical applications as illustrated through a simulation study. \add{As a further application not discussed in this work, the new method appears to be well suited to determine level set confidence regions discussed in \cite{CC09} and \cite{MTN21}.} \add{When $m$ is too large for exact methods to be feasible, other methods may be used to approximate exact $p$-values as hinted at in Appendix \ref{App:Comparison}. Such an approach may be added to the \texttt{ExactMultinom} package in a future version.}

Regarding the choice of test statistic, the ``exact multinomial test'' was treated as a test statistic and the asymptotic distribution of the resulting probability mass statistic was derived. Like most prominent test statistics, the probability mass statistic yields unbiased tests for the uniform null hypothesis. It was shown that a randomized test based on the probability mass statistic can be characterized in that it minimizes the respective (weighted) acceptance region.

Although asymptotic approximations work well in many use cases, there are cases, where these approximations are not adequate, for example, when dealing with small sample sizes or small expectations. On the other hand, there is nothing to be said against the use of exact tests whenever feasible, and it is recommended in the applied literature \citep[p.\ 83]{McD09} for samples of moderate size up to 1000. As the available implementations of exact multinomial tests in R use full enumeration, the new implementation increases the scope of exact multinomial tests for practitioners.

%\bigskip
%\begin{center}
%{\large\bf SUPPLEMENTARY MATERIAL}
%\end{center}
%
%\begin{description}
%
%%\item[Title:] Brief description. (file type)
%
%\item[Appendices:] Mathematical details complementing the proofs of Theorem \ref{Thm:AsApprox}, Proposition \ref{Prop:QMC} and Proposition \ref{Prop:GrowthAR}, \add{and a short comparison with other methods}. (pdf)
%
%\item[R-package ExactMultinom:] R-package containing the implementation described in Section \ref{Sec:3}. (GNU zipped tar file) 
%
%\item[Additional R-code:] R-code used for the simulation study in Section \ref{Sec:4}. (.R file)
%
%\end{description}

%\newpage
\appendix
\section{Difference Between Log-Likelihood Ratio and Probability Mass Statistic}\label{App:A}

\begin{Lemma}
	Let $\pi \in \Deltam$ with $\pi_j > 0$ for all $j = 1,\dots,m$ and $x \in \Deltamn$. Then
	$$T^\P(x,\pi) - T^{G}(x,\pi) = \sum_{j = 1}^m \left(\log(x_j) + 2 r(x_j) - \log(n\pi_j) - 2 r(n\pi_j)\right)$$ for a function $r$ on the positive real numbers for which $0 < r(x) < \frac{1}{12x}$ for $x > 0$. In case $x_j = 0$ for some $j = 1,\dots,m$, the above equality holds if $\log(0) + 2r(0)$ is understood to be 0.
	\label{Lemma:DiffLLRProb}
\end{Lemma}
\begin{proof}
	The logarithm of the Gamma function can be written as 
	$$\log\Gamma(x+1) = \log x\Gamma(x) = x\log(x) - x + \frac{1}{2} \log (2\tilde\pi x) + r(x)$$
	for a function $r$ on the positive real numbers for which $0 < r(x) < \frac{1}{12x}$ holds for all $x > 0$ \citep[see][6.1.41 and 6.1.42; here $\tilde\pi$ denotes Archimedes' constant]{AS72}. This yields
	\begin{align*}
	\log \bar{f}_{n,\frac yn}(y) &= \log \Gamma(n+1) + \sum_j \left(y_j\log \frac{y_j}{n} - \log\Gamma(y_j+1)\right) \\
	&= \log \Gamma(n+1) + \sum_j \left(y_j\log \frac{y_j}{n} - y_j\log(y_j) + y_j - \frac{1}{2}\log (2\tilde\pi y_j) - r(y_j)\right) \\
	&= \log \Gamma(n+1) + n(1-\log n) - \sum_j \left(\frac{1}{2}\log (2\tilde\pi y_j) + r(y_j)\right)
	\end{align*}
	for $y \in \R^m_{>0}$ such that $\sum_j y_j = n$, and hence
	\begin{align*}
	T^\P(x,\pi) - T^{G}(x,\pi)
	&= 2(\log \bar{f}_{n,\pi}(n\pi) - \log f_{n,\frac{x}{n}}(x)) \\
	&= 2 \sum_j \left(\frac{1}{2} \log \frac{x_j}{n\pi_j} + r(x_j) - r(n\pi_j)\right)
	\end{align*}	
\end{proof}

\section{Proof of Proposition \ref{Prop:QMC} b)}\label{App:B}

\begin{proof}
	Throughout the proof, let $x,y \in \Deltamn$ such that $x\neq y$, and define the index sets 
	$$S^+ := \{i\mid x_i > y_i\} \quad \text{ and } \quad S^- := \{j\mid x_j < y_j\}.$$
	
	Let $T = T^\lambda$ and w.l.o.g.\ $T(x) \geq T(y)$. First, consider the case $\lambda > 0$.
	Note that
	\begin{equation}
	T(x) - T(y) = \frac{2}{\lambda(\lambda+1)} \left(\sum_{i \in S^+} \frac{x_i^{\lambda +1} - y_i^{\lambda+1}}{(n\pi_i)^\lambda} - \sum_{j \in S^-} \frac{y_j^{\lambda +1} - x_j^{\lambda+1}}{(n\pi_j)^\lambda} \right) \geq 0
	\label{Eq:DiffPowDivFam}
	\end{equation}
	and		
	\begin{equation}
	\begin{split}
	T(x-e_{i^*}+e_{j^*}) %&= \frac{2}{\lambda(\lambda + 1)} \left(\sum_{j \neq i^*,j^*} x_j \left(\left(\frac{x_j}{n\pi_j}\right)^\lambda - 1\right) + (x_{i^*}-1) \left(\left(\frac{x_{i^*}-1}{n\pi_{i^*}}\right)^\lambda - 1\right) + (x_{j^*} + 1) \left(\left(\frac{x_{j^*}+1}{n\pi_{j^*}}\right)^\lambda - 1\right) \right) \\ 
	= T(x) &- \frac{2}{\lambda(\lambda + 1)}\left(\frac{x_{i^*}^{\lambda +1} - (x_{i^*} - 1)^{\lambda+1}}{(n\pi_{i^*})^\lambda}\right) \\
	&+ \frac{2}{\lambda(\lambda + 1)}\left(\frac{(x_{j^*}+1)^{\lambda +1} - x_{j^*}^{\lambda+1}}{(n\pi_{j^*})^\lambda}\right)
	\end{split}
	\label{Eq:OneStep}
	\end{equation}
	for $i^* \in S^+,j^*\in S^-$. If
	$$i^* = \argmax_{i \in S^+} \frac{x_i^{\lambda +1} - (x_i - 1)^{\lambda+1}}{(n\pi_i)^\lambda},\quad 
	j^* = \argmin_{j \in S^-} \frac{(x_j+1)^{\lambda +1} - x_j^{\lambda+1}}{(n\pi_j)^\lambda}$$
	and $d = d(x,y)$%\frac{1}{2} \Vert x-y\Vert_1$
	, then
	%		\begin{equation}
	%		\begin{split}
	\begin{align*}
	\frac{x_{i^*}^{\lambda +1} - (x_{i^*} - 1)^{\lambda+1}}{(n\pi_{i^*})^\lambda}
	&= \frac{1}{d} \sum_{i \in S^+}\sum_{k = 1}^{x_i-y_i} \frac{x_{i^*}^{\lambda +1} - (x_{i^*} - 1)^{\lambda+1}}{(n\pi_{i^*})^\lambda} \\
	&\geq \frac{1}{d} \sum_{i \in S^+}\sum_{k = 1}^{x_i-y_i} \frac{x_i^{\lambda +1} - (x_i - 1)^{\lambda+1}}{(n\pi_i)^\lambda} \\
	&\geq \frac{1}{d} \sum_{i \in S^+}\sum_{k = 1}^{x_i-y_i} \frac{(x_i+1-k)^{\lambda +1} - (x_i - k)^{\lambda+1}}{(n\pi_i)^\lambda} \\
	&= \frac{1}{d} \sum_{i \in S^+} \frac{x_i^{\lambda +1} - y_i^{\lambda+1}}{(n\pi_i)^\lambda} \\ 
	&\overset{(\ref{Eq:DiffPowDivFam})}{\geq} \frac{1}{d} \sum_{j \in S^-} \frac{y_j^{\lambda +1} - x_j^{\lambda+1}}{(n\pi_j)^\lambda} \stepcounter{equation}\tag{\theequation}\label{Eq:MaxGeqMin} \\
	&= \frac{1}{d} \sum_{j \in S^-}\sum_{k = 1}^{y_j-x_j} \frac{(x_j+k)^{\lambda +1} - (x_j - 1 + k)^{\lambda+1}}{(n\pi_j)^\lambda} \\
	&\geq \frac{1}{d} \sum_{j \in S^-}\sum_{k = 1}^{y_j - x_j} \frac{(x_j+1)^{\lambda +1} - x_j^{\lambda+1}}{(n\pi_j)^\lambda} \\
	&\geq \frac{1}{d} \sum_{j \in S^-}\sum_{k = 1}^{y_j - x_j} \frac{(x_{j^*}+1)^{\lambda +1} - x_{j^*}^{\lambda+1}}{(n\pi_{j^*})^\lambda} \\
	&= \frac{(x_{j^*}+1)^{\lambda +1} - x_{j^*}^{\lambda+1}}{(n\pi_{j^*})^\lambda},
	\end{align*}
	%		\end{split}
	%		\label{Eq:MaxGeqMin}
	%		\end{equation}
	Hence, $T(x) \geq T(x-e_{i^*}+e_{j^*})$ by equation (\ref{Eq:OneStep}).
	
	For $\lambda = 0$, simply taking the limit (as $\lambda \rightarrow 0$) in the above equations with
	\begin{align*}
	i^* &= \argmax_{i \in S^+}\quad 2x_i\log\left(\frac{x_i}{n\pi_i}\right) - 2(x_i-1)\log\left(\frac{x_i -1}{n\pi_i}\right), \\
	j^* &= \argmin_{j \in S^-}\quad 2(x_j +1)\log\left(\frac{x_j+1}{n\pi_j}\right)- 2x_j\log\left(\frac{x_j}{n\pi_j}\right)
	\end{align*}
	yields the desired inequality, since
	\begin{align*}
	&2x_{i^*}\log\left(\frac{x_{i^*}}{n\pi_{i^*}}\right) - 2(x_{i^*}-1)\log\left(\frac{x_{i^*} -1}{n\pi_{i^*}}\right) \\
	&= \lim_{\lambda \rightarrow 0} \frac{2}{\lambda(\lambda + 1)}x_{i^*}\left(\left(\frac{x_{i^*}}{n\pi_{i^*}}\right)^\lambda - 1 \right) \\ 
	&\phantom{=} - \lim_{\lambda \rightarrow 0} \frac{2}{\lambda(\lambda + 1)}(x_{i^*}-1)\left(\left(\frac{x_{i^*}-1}{n\pi_{i^*}}\right)^\lambda - 1\right) \\
	&= \lim_{\lambda \rightarrow 0}\frac{2}{\lambda(\lambda + 1)} \left(\frac{x_{i^*}^{\lambda +1} - (x_{i^*} - 1)^{\lambda+1}}{(n\pi_{i^*})^\lambda} - 1\right) \\
	&\overset{(\ref{Eq:MaxGeqMin})}{\geq} \lim_{\lambda \rightarrow 0}\frac{2}{\lambda(\lambda + 1)} \left(\frac{(x_{j^*} + 1)^{\lambda +1} - x_{j^*}^{\lambda+1}}{(n\pi_{j^*})^\lambda} - 1\right) \\
	&= 2(x_{j^*} +1)\log\left(\frac{x_{j^*}+1}{n\pi_{j^*}}\right)- 2x_{j^*}\log\left(\frac{x_{j^*}}{n\pi_{j^*}}\right).
	\end{align*}
\end{proof}

\section{Details for the Proof of Proposition \ref{Prop:GrowthAR}}\label{App:C}

The following two lemmas provide further details not contained in the proof of Proposition \ref{Prop:GrowthAR} itself. 

\begin{Lemma}
	Using notation as in the proof of Proposition \ref{Prop:GrowthAR}, $x \mapsto \bar{T}(x)$ is convex.
	\label{Lemma:Convex}
\end{Lemma}

\begin{proof}
	The function $x\mapsto\bar{T}^{\chi^2}(x) = \sum_j \frac{x_j^2}{n\pi_j} - n$ is clearly convex as it is a sum of convex functions.
	
	The function $x\mapsto\bar{T}^{G}(x) = 2 \sum_j x_j\log(x_j) - x_j \log(n\pi_j)$ is convex\add{,} since $x \mapsto x\log(x)$ is convex (an elementary proof of this can be given using either the inequality of the arithmetic and geometric means or the second derivative).
	
	The function $x \mapsto\bar{T}^\P(x) = 2(\log(\bar{f}_{n,\pi}(n\pi)) - \log(\Gamma(n+1)) + \sum_j \log(\Gamma(x_j+1)) - \sum_j x_j \log(p_j))$ is convex as the Gamma function is logarithmically convex by the Bohr-Mollerup theorem \citep[Theorem 2.4.2]{BW10}.
\end{proof}

\begin{Lemma}
	Using notation as in the proof of Proposition \ref{Prop:GrowthAR}, the function \allowbreak $\partial B_{r_0}(\pi) \rightarrow \R, x_0\mapsto\bar{T}(x(n,x_0))$ converges uniformly to $\bar{T}^{\chi^2}(x(n,x_0))$ as $n\rightarrow\infty$ if $T = T^{G}$ or $T = T^\P$.
	\label{Lemma:UnifConv}
\end{Lemma}

\begin{proof}
	Let $x_0 \in \partial B_{r_0}(\pi)$, and define $c = c(x_0) := \sqrt{n_0}(x_0 - \pi)$. Hence $\vert c_j \vert \leq \sqrt{n_0}r_0 < \sqrt{n_0}$ for all $j = 1,\dots,m$. Consider first the case $T = T^{G}$.
	Then (using the Taylor expansion $\log(1+x) = \sum_{k = 1}^\infty (-1)^{k+1}\frac{x^k}{k}$)
	\begin{align*}
	\bar{T}(x(n,x_0)) &= 2\sum_{j=1}^m x(n,x_0)_j\log\frac{x(n,x_0)_j}{n\pi_j} \\
	&= 2\sum_j(n\pi_j + \sqrt n c_j)\log\frac{n\pi_j + \sqrt nc_j}{n\pi_j} \\
	&= 2\sum_j(n\pi_j + \sqrt n c_j) \sum_{k = 1}^\infty\frac{(-1)^{k+1}}k\left(\frac{c_j}{\sqrt n\pi_j}\right)^k \\
	&= 2\sum_j \Bigg(\sqrt n c_j + \frac{c_j^2}{2\pi_j} - \frac{c_j^3}{2\sqrt n \pi_j^2} + \frac{n\pi_j + \sqrt n c_j}{\sqrt n^3}\sum_{k = 3}^\infty\frac{(-1)^{k+1}c_j^k}{k\sqrt n^{k-3}\pi_j^k}\Bigg)
	%\sum_{k = 3}^\infty\frac{(-1)^{k+1}}k\left(\frac{c_j^k}{\sqrt n ^{k-2}\pi_j^{k-1}} + \frac{c_j^{k+1}}{\sqrt n^{k-1}\pi_j^k}\right)
	\end{align*}
	As $\sum_j c_j = 0$ and $2\sum_j \frac{c_j^2}{2\pi_j} = T^{\chi^2}(x(n,x_0))$, the inequalities
	\begin{align*}
	&\vert \bar{T}^{\chi^2}(x(n,x_0))-\bar{T}(x(n,x_0))\vert \\
	&< \sum_j\left(\frac{\vert c_j\vert^3}{2\sqrt n \pi_j^2} + \frac{n\pi_j + \sqrt n \vert c_j\vert}{\sqrt n^3}\sum_{k = 3}^\infty\frac{\vert c_j\vert^k}{k\sqrt n^{k-3}\pi_j^k}\right) \\
	&< \sum_j\left(\frac{\sqrt{n_0}^3}{2\sqrt n \pi_j^2} + \frac{n\pi_j + \sqrt n \sqrt{n_0}}{\sqrt n^3}\sum_{k = 3}^\infty\frac{\sqrt{n_0}^k}{k\sqrt n^{k-3}\pi_j^k}\right) \\
	&< \frac{1}{\sqrt n}\sum_j\left(\frac{\sqrt{n_0}^3}{2\pi_j^2} + (\pi_j + \sqrt{n_0})C(n)\right)
	\end{align*}
	hold, where the series converges to some $C(n)$ for sufficiently large $n$ by the ratio test and $C(n)$ decreases as $n$ increases. As this upper bound is independent of the choice of $x_0$ uniform convergence is ensured.
	
	Using Lemma \ref{Lemma:DiffLLRProb} in case $T = T^\P$, the inequality
	\begin{align*}
	&\vert \bar T^{G}(x(n,x_0))-\bar T(x(n,x_0))\vert \\
	&= \left\vert \sum_{j = 1}^m \left(\log \frac{x(n,x_0)_j}{n\pi_j} + 2 r(x(n,x_0)_j) - 2 r(n\pi_j)\right)\right\vert \\
	&= \left\vert\sum_j \left(\log\frac{n\pi_j + \sqrt{n}c_j}{n\pi_j} + 2r(n\pi_j+ \sqrt{n}c_j) - 2r(n\pi_j)\right)\right\vert \\
	&< \sum_j\left(\left\vert \log\left(1 - \frac{\sqrt{n_0}r_0}{\sqrt{n}\pi_j}\right)\right\vert + \frac{2}{12(n\pi_j - \sqrt{nn_0}r_0)}\right)\\
	\end{align*}
	holds and the upper bound converges to zero independent of the choice of $x_0$. Hence
	$$\bar{T}^{\chi^2}-\bar{T} = (\bar{T}^{\chi^2}-\bar{T}^{G}) + (\bar{T}^{G}-\bar{T})$$
	%$$\bar{T}^{\chi^2}(x(n,x_0))-\bar{T}(x(n,x_0)) = (\bar{T}^{\chi^2}(x(n,x_0))-\bar{T}^{G}(x(n,x_0))) + (\bar{T}^{G}(x(n,x_0))-\bar{T}(x(n,x_0)))$$
	converges uniformly to zero as a function on $\partial B_{r_0}(\pi)$ in the sense of the lemma.
\end{proof}

\add{
\section{Comparison with Other Methods} \label{App:Comparison}

As hinted at in the introduction, other approaches for computing exact multinomial $p$-values exist. However, none of these methods have considered the probability mass statistic, but have focused on the log-likelihood ratio statistic \citep{Rah03,KN06} and other statistics from the family of power divergence statistics \citep{BOP92,Hir97,BFT04}. Adaptions of these methods to the probability mass statistic are beyond the scope of the present work.

\begin{table}[t]
	\caption{\add{Runtime and $p$-values obtained by different methods for the five pairs from Table \ref{Tab:RandValues}. Results from the full enumeration implemented by \texttt{xmulti} were included to show agreement of $p$-values produced by the exact methods. \emph{Branch \& Bound} refers to the implementation by \cite{Bej06} and \emph{Dynamic} refers to the dynamic programming approach by \cite{Rah03} as implemented by the author with lattice size $q$. Times are in milliseconds.}}
	\centering\small
	\begin{tabular}{cc|cc|cc|cc|cc}
		\multicolumn{2}{c|}{Algorithm \ref{Algo:P-val}} & \multicolumn{2}{c|}{Branch \& Bound} & \multicolumn{2}{c|}{\texttt{xmulti}} & \multicolumn{2}{c|}{Dynamic ($q = 10^3$)} & \multicolumn{2}{c}{Dynamic ($q = 10^4$)} \\
		$p_{T^{G}}$ & time & $p_{T^{G}}$ & time & $p_{T^{G}}$ & time & $p_{T^{G}}$ & time & $p_{T^{G}}$ & time \\
		\hline
		0.0126 & 1.6 & 0.0126 & 2.7 & 0.0126 & 29.8 & 0.0141 & 22.2 & 0.0135 & 240.2 \\
		0.0361 & 3.5 & 0.0361 & 6.7 & 0.0361 & 29.1 & 0.0339 & 22.0 & 0.0359 & 237.2 \\
		0.0719 & 1.6 & 0.0719 & 5.8 & 0.0719 & 28.9 & 0.0675 & 21.2 & 0.0721 & 224.4 \\
		0.0461 & 0.9 & 0.0461 & 2.3 & 0.0461 & 29.3 & 0.0758 & 22.2 & 0.0460 & 241.4 \\
		0.0628 & 1.7 & 0.0628 & 5.0 & 0.0628 & 29.2 & 0.0967 & 21.8 & 0.0625 & 235.5 \\
	\end{tabular}
	\label{Tab:RandValues2}
\end{table}

Most other methods are not ``strictly exact'' but compute the distribution of a discretized test statistic under the null hypothesis \citep{KN06}, thereby reducing the complexity of the resulting algorithms to polynomial time regardless of the number of \add{categories} $m$. While this seems to result in good approximations of very small $p$-values, which are of interest in some bioinformatics applications, the approximations are not exact and may differ quite strongly from the exact $p$-values of moderate size depending on the granularity of the discretization  (see Table \ref{Tab:RandValues2}). This seems to be amplified by the fact that test statistic values span quite a large range, but most of the probability mass is concentrated in a small part of this range. Of course, using finer discretizations improves these approximations, however, increasing the lattice size (i.e., the number of discretized values of the test statistic) increases the runtime (and memory usage) in practice. An instructive mathematical formulation of the idea as a dynamic programming problem is given by \cite{Rah03}, which was implemented by the author to obtain the results in Table \ref{Tab:RandValues2}. This approach has complexity of $\mathcal{O}(mqn^2)$, where the lattice size $q \in \N$ needs to grow linearly with $n$ to preserve the accuracy of the approximation.  The approach by \cite{KN06} reduces the complexity to $\mathcal{O}(mqn\log(n))$ (for the log-likelihood ratio statistic) by using a discrete Fourier transform to obtain the distribution of the discretized test statistic. Nonetheless, as these approaches allow to approximate exact $p$-values when $m$ is too large for exact algorithms to be feasible, such an approach may be added to the \texttt{ExactMultinom} package in a future version.

\begin{figure}[t]\centering
	\includegraphics[scale = 0.5]{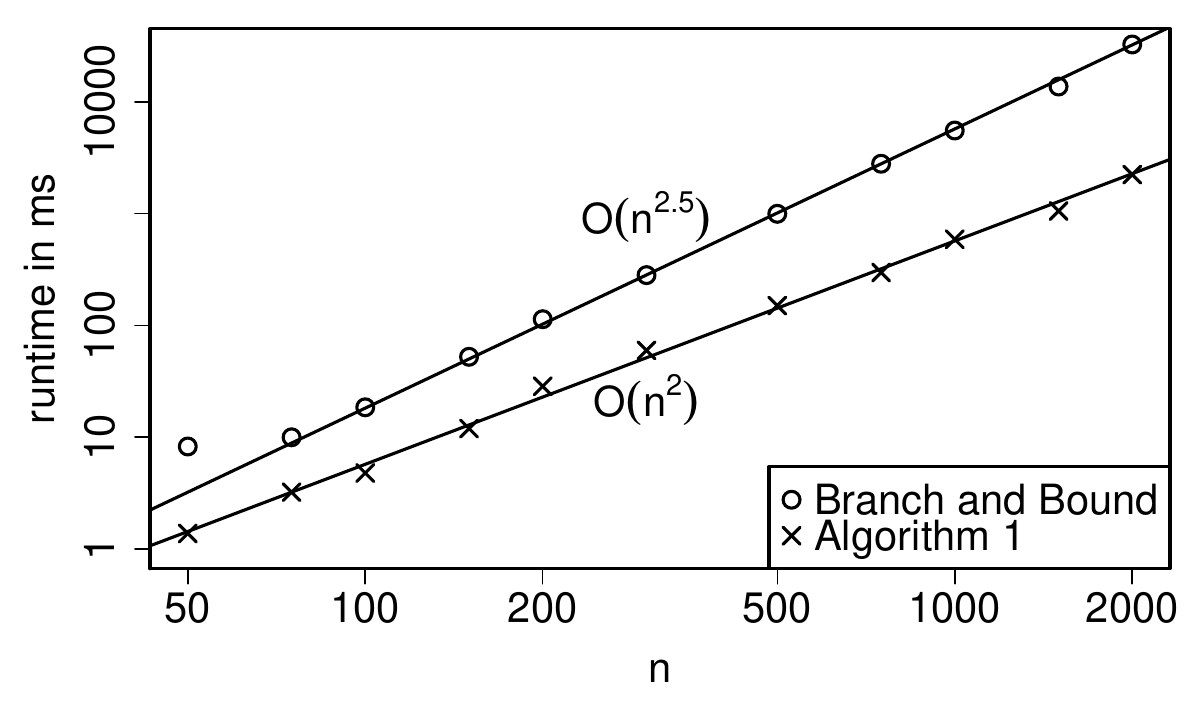}
	\caption{\add{Runtime of the branch and bound approach by \cite{BFT04} and the implementation of Algorithm \ref{Algo:P-val} for random samples with $p$-values of about 0.001 and null hypothesis $\pi = (\frac1{15},\frac2{15},\frac3{15},\frac4{15},\frac5{15})$.}}
	\label{Fig:Time5_App}
\end{figure} 

The only exact approach is the one by \cite{BFT04} implemented by \cite{Bej06}. \cite{BFT04} employ a ``branch and bound'' approach to speed up the computation of exact multinomial $p$-values. However, this approach also suffers from exponential runtime in $m$. The implementation by \cite{Bej06} computes exact $p$-values for the log-likelihood ratio statistic and can be adapted to any statistic in the family of power divergence statistics. Similar to Algorithm 1, the runtime of the branch and bound approach depends on the null hypothesis parameter $\pi$ and increases as the $p$-value decreases \citep[][Figure 5]{BFT04}. Figure \ref{Fig:Time5_App} shows runtime as a function of $n$ for $m = 5$ for random samples with $p$-values of about 0.001. Clearly, the implementation of Algorithm \ref{Algo:P-val} discussed in the main paper outperforms the implementation by \cite{Bej06} in this exemplary run, even though the former computes $p$-values for multiple test statistics at once. Figure \ref{Fig:Time5_App} suggests that the branch and bound approach may have complexity of $\mathcal{O}(n^{\frac{m}{2}})$ (in agreement with Figure 4 in \cite{BFT04}). Adapting the branch and bound approach to the probability mass statistic is left as a subject for future research.
}

\bibliographystyle{chicago}
\bibliography{manuscript_arxiv}

\end{document}